\documentclass[11pt]{article}

\usepackage[margin=1in]{geometry}

\usepackage{amsmath, amssymb, amsthm}
\usepackage[mathic]{mathtools}

\usepackage[T1]{fontenc}
\usepackage{libertine}
\usepackage[libertine]{newtxmath} 
\usepackage[scaled=0.96]{zi4} 

\usepackage{a4wide, dsfont, microtype, xcolor, paralist}

\usepackage{graphicx}
\usepackage{amsfonts}
\usepackage{hyperref}

\usepackage{changepage}

\usepackage{listofitems}

\usepackage{float}

\usepackage{thmtools} 

\usepackage{array}
\newcolumntype{P}[1]{>{\footnotesize\centering\arraybackslash}p{#1}}
\newcommand{\PreserveBackslash}[1]{\let\temp=\\#1\let\\=\temp}
\newcolumntype{C}[1]{>{\footnotesize\PreserveBackslash\centering}p{#1}}

\usepackage{tikz}
\usetikzlibrary{positioning}
\usetikzlibrary{arrows.meta}
\usetikzlibrary{automata}
\usetikzlibrary{matrix}
\usetikzlibrary{calc}
\usetikzlibrary{math} 

\usepackage{subcaption}

\usepackage{cite}

\DeclareMathOperator{\ones}{ones}

\usepackage{multirow}

\newcommand{\DD}[0]{\ensuremath{D}}
\newcommand{\LL}[0]{\ensuremath{L}}
\newcommand{\wc}[0]{wc}

\newcommand{\BWT}[0]{XBWT}
\newcommand{\middlealg}[0]{mid}

\newcommand{\valnumber}[0]{val}
\newcommand{\countop}[0]{count}
\newcommand{\parentop}[0]{parent}

\newcommand{\select}[0]{select}
\newcommand{\prank}[0]{p\_rank}
\newcommand{\rank}[0]{rank}

\let\epsilon\varepsilon

\usepackage{xspace}

\usepackage[linesnumbered,ruled,vlined,boxed]{algorithm2e}
\newlength{\commentWidth}
\let\oldnl\nl
\newcommand{\nonl}{\renewcommand{\nl}{\let\nl\oldnl}}
\DontPrintSemicolon
\SetKwInOut{Input}{Input}\SetKwInOut{Output}{Output}

\usepackage{tikz}
\usetikzlibrary{decorations.pathreplacing}
\usetikzlibrary{plotmarks}
\usetikzlibrary{positioning,automata,arrows}
\usetikzlibrary{shapes.geometric}
\usetikzlibrary{decorations.markings}
\usetikzlibrary{positioning, shapes, arrows}

\PassOptionsToPackage{usenames,dvipsnames,svgnames}{xcolor}
\definecolor{orange}{RGB}{235,90,0}
\definecolor{darkorange}{RGB}{175,30,0}
\definecolor{turkis}{RGB}{131,182,182}
\definecolor{darkturkis}{RGB}{31,82,82}
\definecolor{green}{RGB}{102,180,0}
\definecolor{darkgreen}{RGB}{51,90,0}
\definecolor{myblue}{RGB}{0,0,213}
\definecolor{mydarkblue}{RGB}{0,0,100}
\definecolor{mybrightblue}{HTML}{74B0E4}
\definecolor{mybrighterblue}{HTML}{B3EAFA}
\definecolor{lila}{RGB}{102,0,102}
\definecolor{darkred}{RGB}{139,0,0}
\definecolor{darkyellow}{RGB}{188,135,2}
\definecolor{brightgray}{RGB}{200,200,200}
\definecolor{darkgray}{RGB}{50,50,50}
\definecolor{amaranth}{rgb}{0.9, 0.17, 0.31}
\definecolor{alizarin}{rgb}{0.82, 0.1, 0.26}
\definecolor{amber}{rgb}{1.0, 0.75, 0.0}
\definecolor{green(ryb)}{rgb}{0.4, 0.69, 0.2}
\definecolor{hanblue}{rgb}{0.27, 0.42, 0.81}
\definecolor{grannysmithapple}{rgb}{0.66, 0.89, 0.63}

\newtheorem{theorem}{Theorem}[section]
\newtheorem{lemma}[theorem]{Lemma}
\newtheorem{definition}[theorem]{Definition}
\newtheorem{corollary}[theorem]{Corollary}

\newtheorem{remark}[theorem]{Remark}
\newtheorem{proposition}[theorem]{Proposition}

\newenvironment{acknowledgements}{%
  \begin{abstract}
}{%
  \end{abstract}
}

\title{New Entropy Measures for Tries with Applications to the XBWT}
\author{
 	Lorenzo Carfagna\thanks{Department of Computer Science, University of Pisa, Italy} \\ \texttt{lorenzo.carfagna@phd.unipi.it} \and
    Carlo Tosoni\thanks{DAIS, Ca' Foscari University of Venice, Italy} \\
    \texttt{carlo.tosoni@unive.it}
}
\date{}

\begin{document}

\maketitle              
\begin{abstract}
Entropy quantifies the number of bits required to store objects under certain given assumptions.
While this is a well established concept for strings, in the context of tries the state-of-the-art regarding entropies is less developed.
The standard trie worst-case entropy considers the set of tries with a fixed number of nodes and alphabet size.
However, this approach does not consider the frequencies of the symbols in the trie, thus failing to capture the compressibility of tries with skewed character distributions. 
On the other hand, the label entropy [FOCS '05], proposed for node-labeled trees, does not take into account the tree topology, which has to be stored separately.
In this paper, we introduce two new entropy measures for tries - worst-case and empirical - which overcome the two aforementioned limitations.
Notably, our entropies satisfy similar properties of their string counterparts, thereby becoming very natural generalizations of the (simpler) string case. 
Indeed, our empirical entropy is closely related to the worst-case entropy and is reachable through a natural extension of arithmetic coding from strings to tries.
Moreover we show that, similarly to the FM-index for strings [JACM '05], the XBWT of a trie can be compressed and efficiently indexed within our $k$-th order empirical entropy plus $o(n)$ bits, with $n$ being the number of nodes.
Interestingly, the space usage of this encoding includes the trie topology and the upper-bound holds for every $k$ sufficiently small, simultaneously.
This XBWT encoding is always strictly smaller than the original one [JACM '09] and we show that in certain cases it is \emph{asymptotically} smaller.
\end{abstract}

\begin{acknowledgements}
We would like to thank Sung-Hwan Kim for remarkable suggestions concerning the combinatorial problem addressed in the article.
Part of this research work was conducted while the authors were visiting the Department of Computer Science of the University of Helsinki.

\emph{Lorenzo Carfagna}: partially funded by the INdAM-GNCS Project CUP E53C24001950001, 
and by the PNRR ECS00000017 Tuscany Health Ecosystem, Spoke 6, CUP I53C22000780001, funded by the NextGeneration EU programme.

\emph{Carlo Tosoni}: funded by the European Union (ERC, REGINDEX, 101039208). Views and opinions expressed are however those of the authors only and do not necessarily reflect those of the European Union or the European Research Council Executive Agency. Neither the European Union nor the granting authority can be held responsible for them.
\end{acknowledgements}

{\let\thefootnote\relax\footnotetext{A preliminary and partial version of this article appeared in the \emph{Proceedings of the 32nd International Symposium on String Processing and Information Retrieval (SPIRE 2025)}, pp. 18–27, 2025.}}

\section{Introduction}

A \emph{cardinal tree} is an ordered tree where every node is conceptually represented by an array of $\sigma$ positions, and every position in that array is either null or stores a pointer to a child node.
Cardinal trees are equivalent to tries, that is, edge-labeled ordered trees over an alphabet of size $\sigma$, in which (i)~the labels of the edges outgoing from the same node are distinct, and (ii)~sibling nodes are ordered by their incoming label.
Tries are widely studied in the literature and are regarded as a central object in computer science, since they can be applied in a wide range of applications.
Indeed, tries can be used to represent a dictionary of strings while supporting important operations such as dictionary and prefix matching~\cite{AC75, SuffixTrees, originalArticleTries}.
For these reasons, several compact implementations for tries have been proposed in the literature.
The space complexity of many of these representations has been expressed in terms of the worst-case entropy $\mathcal{C}(n,\sigma)$ defined as $\log_2\frac{1}{n}\binom{n\sigma}{n-1}$, where $\frac{1}{n}\binom{n\sigma}{n-1}$ is the number of tries with $n$ nodes over an alphabet of size $\sigma$~\cite{cycleLemma, concreteMathematics}. 
For instance, Benoit et al.\ proposed a trie representation taking $\mathcal{C}(n,\sigma) + \Omega(n)$ bits~\cite{RepresentingTrees}.
This space occupation was later improved to $\mathcal{C}(n, \sigma) + o(n) + O(\log \log \sigma)$ bits by Raman et al.~\cite{RRR} and Farzan et al.~\cite{UniversalSuccinctTree?}.
All these static data structures were devised to support queries on the trie topology plus the cardinal query of retrieving the child labeled with the $i$-th character of the alphabet in constant time.
Moreover, researchers also proposed dynamic representations for tries~\cite{ArroyueloDynamicTries, Bonsai, DavoodiDynamicTries} with $\mathcal{C}(n,\sigma) + \Omega(n)$ space complexity, in bits.
The function $\mathcal{C}(n,\sigma)$ is a lower-bound for the number of bits required in the worst-case to represent a trie with $n$ nodes and alphabet size $\sigma$, however, $\mathcal{C}(n,\sigma)$ only exploits information about these parameters.  
As a consequence, this worst-case entropy does not take advantage of any information regarding the distribution of the edge labels in the trie, and therefore it fails to capture the compressibility of tries with skewed character distributions.
On the contrary, the situation is different when it comes to strings.
Indeed, the (string) empirical entropy is a widely used and established compressibility measure accounting also for the frequencies of the symbols in the input string~\cite{compactDataStructures}.
Thus, if the distribution of the characters in a string $S$ is skewed, then a statistical compressor, such as Huffman coding~\cite{HuffmanCoding} and arithmetic coding~\cite{magicAlgorithms}, can represent $S$ in a number of bits significantly smaller than $n\log \sigma$ bits, which is the string counterpart of the worst-case formula $\mathcal C(n,\sigma)$. 
Moreover, there exist several compressed indices for strings whose space complexity is expressed in terms of the $k$-th order empirical entropy of the input string~\cite{FullTextIndexesSurvey}. 
In particular, the famous FM-index~\cite{FMindex,FMI-Journal} compresses the Burrows-Wheeler transform~\cite{BWT} of a string in a space close to its high order empirical entropy, while at the same time supporting efficient indexing operations directly on the compressed string.
Later, a similar scenario emerged in the context of trees.
Indeed, Ferragina et al.~\cite{xBWTjournal} extended the original Burrows-Wheeler transform to node-labeled ordered trees, where this extension is known in the literature as XBWT.
The XBWT sorts the tree nodes based on the string reaching them from the root, and it represents the input tree with two coordinated arrays, accounting respectively for the labels and the topology of the tree.
Similarly as happens for strings, this first array can be compressed to a $k$-th order empirical (tree) entropy known as the label entropy $\mathcal{H}^{label}_k$~\cite{xBWTconf,xBWTjournal}.
In this setting, the $k$-length context of a label is the $k$-length string incoming to its corresponding tree node.
However, this label entropy considers only the information about the distribution of the symbols, without accounting for the tree topology.
The label entropy has also been applied to the more specific context of tries.
Indeed, other works~\cite{HonTrieLabelEntropy,Kosolobov2019} proposed compressed representations of the XBWT, specifically adapted for tries, and analysed their space usage in terms of the label entropy.
However, since the label entropy has originally been proposed for general trees, this measure does not capture the following intrinsic characteristic property of tries: the label distribution imposes constraints on the underlying trie topology.
Indeed, on tries, sibling nodes cannot have the same incoming label.
As a consequence of this, the more skewed the symbol distribution is, the more constrained the trie structure becomes.
For example, tries over a unary alphabet have a fixed shape.
For these reasons, we aim to introduce new trie entropy measures, which 1)~contrary to the label entropy, are specifically designed for tries and consider also the trie structure, and 2)~refine the previous worst-case entropy $\mathcal C(n,\sigma)$ by exploiting also the distribution of the edge symbols.

\paragraph{Our contribution.} Motivated by the above reasons, in this paper we study the combinatorial problem of finding the number $\lvert \mathcal{U} \rvert$ of tries having a given \emph{symbol distribution} $\{n_c\mid c\in \Sigma\}$, that is the tries having $n_c$ edges labeled by the symbol $c$ drawn for the input alphabet $\Sigma$.
During a bibliographic search, we found a recent technical report~\cite{TRProdinger} showing $\lvert \mathcal{U}\rvert = \frac{1}{n}\prod_{c \in \Sigma}\binom{n}{n_c}$, with a proof based on generating functions~\cite{analyticCombinatorics}.
In this paper we prove this closed formula by showing a simple bijection between these tries and a particular class of binary matrices.
As an immediate consequence, we obtain a worst-case entropy $\mathcal{H}^{\wc}=\log_2 \lvert \mathcal{U}\rvert$ which takes into account the distribution of the trie labels.
Since $\mathcal{H}^{\wc}$ considers a particular subset of tries, it is easy to observe that $\mathcal{H}^{\wc}$ refines the previous entropy $\mathcal C(n,\sigma)$. 
In addition, $\mathcal{H}^{\wc}$ can be significantly smaller if the symbol distribution is skewed.
We also define a notion of $k$-th order empirical entropy $\mathcal H_k$, specific for tries.
Similarly to the label entropy $\mathcal{H}^{label}_k(\mathcal{T})$, our empirical measure $\mathcal{H}_k$ divides the nodes based on their $k$-length context, but it also encodes the topology of the trie.
Moreover, we emphasise that these new entropies $\mathcal{H}^{\wc}$ and $\mathcal{H}_{k}$ share several properties with their string counterparts, and that they can be consequently regarded as their natural extensions in the domain of tries.
For instance, we show that similarly to strings, it holds that $n\mathcal{H}_0=\mathcal{H}^{\wc}+O(\sigma \log n)$, and for every $k\geq 0$ we have $\mathcal{H}_{k+1} \leq \mathcal{H}_{k}$.
Notably, we show that our empirical entropy is reachable, i.e., given any fixed $k$ we can encode any trie $\mathcal{T}$ with at most $n\mathcal{H}_k(\mathcal{T}) + 2$ bits plus the space $(\sigma + 1)\sigma^k\lceil \log n\rceil$ in bits required to store the (conditional) probabilities of its labels.
We prove this by extending the notorious arithmetic coding from strings to tries.
Regarding the comparison with $\mathcal{H}^{label}_k$, we show that in the context of tries, our $k$-th order empirical entropy is always upper-bounded by this $k$-order label entropy plus an $1.443n$ term.
Interestingly, the $1.443n$ term is smaller than the worst-case entropy $2n - \Theta(\log n)$ of unlabeled ordered trees~\cite{compactDataStructures}, expressing the bits required to store the trie topology separately regarded as an ordered tree of $n$ nodes. 
We also show that there exists an infinite family of tries for which our measure is asymptotically smaller than $\mathcal{H}^{label}_k$ alone.
These results have direct impact on the analysis of known compressed representations for the XBWT of a trie and the Aho-Corasick automaton~\cite{HonTrieLabelEntropy, Kosolobov2019, AC75, BelazzouguiAhoCorasick}.
Indeed, we refine the space analysis of an existing XBWT representation for tries~\cite[Lemma 6]{Kosolobov2019} to $n\mathcal{H}_k(\mathcal{T}) + o(n)$ bits.
This bound assumes $\sigma \leq n^\varepsilon$ and holds simultaneously for every $k \leq \max\{0, \alpha \log _\sigma n - 2 \}$, where $\alpha,\varepsilon$ are any two constants satisfying $0 \leq \alpha,\varepsilon < 1$.
Notably, for tries, this space complexity is strictly smaller than the one used by the original XBWT for trees~\cite[Section 3]{xBWTjournal}.
We also show how this representation can count the number of nodes reached by a string pattern $p$ in $O(\lvert p\rvert (\log\sigma + \log\log n))$ time, where the path can be anchored anywhere in the trie.
For poly-logarithmic alphabets, we propose another XBWT trie representation which improves this time complexity to $O(m)$.
Moreover, we show that if every symbol $c$ occurs at most $n/2$ times, then these representations are \emph{succinct}, that is they take at most $\mathcal{H}^{\wc}+o(\mathcal{H}^{\wc})$ bits of space.
These results prove a strong connection between our trie empirical entropy and the XBWT of a trie.
This connection is analogous to the one existing between the original BWT and the empirical entropy of a string, as these data structures can be regarded as natural extensions of the FM-index for strings to tries~\cite{FMindex, FMI-Journal}.
Finally, we compare our empirical entropy with the trie repetitiveness measure $r$ proposed by Prezza~\cite{rindexTrie}, which counts the number runs in the trie XBWT.
Due to~\cite[Theorem 3.1]{rindexTrie}, we show that for every trie and integer $k \geq 0$ it holds that $r \leq n\mathcal{H}_k + \sigma^{k+1}$.
We also provide an infinite family of tries for which $r = \Theta(n\mathcal{H}_0) = \Theta(n)$ holds.
Also these results are consistent with the ones known for the corresponding string measures~\cite{NavarroRipetitivenessMeasures}.

The rest of the article is organised as follows.
In Section~\ref{sec: Notation} we introduce the notation used in the paper.
In Section~\ref{section counting tries} we address the problem of counting the number of tries with a given symbol distribution.
In Section~\ref{sec:entropy}, we introduce our entropy measures and we analyse their relation. 
We present arithmetic coding for tries and we show the reachability of our higher-order empirical entropy.
Then we compare this measure with the label entropy.
In Section~\ref{section: XBWT}, we recall the XBWT of a trie and we analyse the XBWT representation in~\cite{Kosolobov2019} in terms of our empirical measure.
In Section~\ref{sec: fmindex}, we show how these XBWT representations can efficiently support count queries directly in the compressed format and we prove their succinctness.
In Section~\ref{sec: r index} we compare our $k$-th order entropy with the trie repetitiveness measure $r$.
Finally, in Section~\ref{sec: conclusions}, we provide some concluding remarks summarising the contents of the article.
 
\section{Notation}\label{sec: Notation}

\paragraph{Strings and matrices.}
In the following we refer with $\Sigma$ a finite alphabet of size $\sigma$, totally ordered according to a given relation $\preceq$.
For every integer $k\geq 0$ we define $\Sigma^k$ as the set of all length-$k$ strings with symbols in $\Sigma$ where $\Sigma^0=\{\epsilon\}$, i.e., the singleton set containing the empty string $\epsilon$.
Moreover, we denote by $\Sigma^*$ the set of all finite strings over the alphabet $\Sigma$, i.e., the union of all $\Sigma^k$ for every $k\geq 0$.
We extend $\preceq$ \emph{co-lexicographically} to the set $\Sigma^*$, i.e., by comparing the  characters from right to left and considering $\epsilon$ as the smallest string in $\Sigma^*$.
With $[n]$ we denote the set $\{1,..,n\}\subseteq \mathbb{N}$ and given a set $X$, $\lvert X \rvert$ is the cardinality of $X$; similarly, given a string $S\in \Sigma^k$, then $\lvert S \rvert=k$ is its length.
Unless otherwise specified, all logarithms are base $2$ and denoted by $\log x$, furthermore we assume that $0\log(x/0)=0$ for every $x \geq 0$.
The \emph{zero-th order empirical entropy} $\mathcal{H}_0(S)$ of a string $S\in \Sigma^n$ is defined as $\sum_{c\in \Sigma} (n_c/n) \log (n/n_c)$ where $n_c$ is the number of times the symbol $c$ occurs in $S$~\cite{compactDataStructures}.
This entropy is generalised to higher orders as follows~\cite{compactDataStructures}.
We say that $S[i]$ has context $w\in \Sigma^k$ if $S[i]$ immediately follows an occurrence of $w$ in $S$.
Let $S_w$ be the string obtained by concatenating (in any order) all the characters $S[i]$ having context $w$, then for $k\geq 0$ the $k$-\emph{th order empirical entropy} of $S$ is $\mathcal{H}_k(S)=\sum_{w\in \Sigma^k} (\lvert S_w\rvert/n)\mathcal{H}_0(S_w)$.
Given a matrix $M$ we denote the $i$-th row/column of $M$ as $M[i][-]$ and $M[-][i]$, respectively.
In the following, given a binary matrix $M$, we denote with $\ones(M')$ the number of entries equal to $1$ in a submatrix $M'$ of $M$.
We work in the RAM model, where we assume that the word size is $\Theta(\log n)$, where $n$ is the input dimension.

\paragraph{Tries.}\label{subsection definitions on tries} In this paper, we work with \emph{cardinal trees}, also termed tries, i.e., edge-labeled ordered trees in which (i)~the labels of the edges outgoing from the same node are distinct, and (ii)~sibling nodes are ordered by their incoming label.
We denote by $\mathcal{T} = (V,E)$ a trie over an alphabet $\Sigma$, where $V$ is set of nodes with $\lvert V\rvert=n$ and $E$ is the set of edges with labels drawn from $\Sigma$.
We also denote with $n_c$ the number of edges labeled by the character $c \in \Sigma$.
Note that since $\mathcal{T}$ is a tree we have that $\sum_{c \in \Sigma} n_c = \lvert E \rvert= n-1$.
We denote by $\lambda(u)$ the label of the incoming edge of a node $u \in V$. 
If $u$ is the root, we define $\lambda(u) = \#$, where $\#$ is a new special symbol not labeling any edge, which according to $\preceq$, is smaller than all the other characters of $\Sigma$.
Given a node $u \in V$, the function $\pi(u)$ returns the parent node of $u$ in $\mathcal{T}$, where $\pi(u) = u$ if $u$ is the root.
Furthermore, the function $out(u)$ returns the set of labels of the edges outgoing from node $u$, i.e., $out(u) = \{c \in \Sigma: (u,v,c) \in E\; \text{for some}\; v \in V \}$.
Note that a node $u$ is a leaf if $out(u)=\emptyset$ and is internal otherwise.
We define the \emph{depth} of a node $u$ as the number of edges we traverse to reach $u$ from the root and the height $h$ of a trie $\mathcal{T}$ as the maximum depth among its nodes where by definition the root of $\mathcal{T}$ has depth $0$.

\paragraph{Bitvectors.} Let $B$ be a bitvector of size $m$ with $n$ entries equal to $1$.
We now recall some fundamental operations on bitvectors, namely \emph{rank}, \emph{partial rank}, and \emph{select} queries.
A (full) \emph{rank} query on $B$, denoted by $\rank(i,B)$, returns the number of $1$s in $B$ up to position $i$ (included), with $i \in [m]$.
For every $i \in [m]$ a \emph{partial rank} query, denoted by $\prank(i, B)$, outputs $\rank(i,B)$ if $B[i] = 1$, and $-1$ otherwise.
On the other hand, a \emph{select query}, denoted by $\select(i,B)$, returns the position of the $i$-th occurrence of $1$ in $B$, where $i \in [n]$.
An \emph{indexable dictionary} (ID) representation for $B$ is a data structure supporting \emph{select} and \emph{partial rank} queries on $B$ in constant time. If an ID representation also supports (full) rank queries in $O(1)$ time, then we call it a \emph{fully} indexable dictionary (FID) representation\footnote{In some articles, FIDs are required to be able to perform rank and select in $O(1)$ time also on the complementary bitvector (see for instance~\cite{RRR}).
However, in this paper we are not interested in these operations.}.
Note that both rank and partial rank queries can be used to determine the value of $B[i]$ for every $i\in [m]$.
We also recall that an ID representation can easily support full rank queries in logarithmic time by means of a binary search on $B$ as follows.
Suppose we want to answer the query $\rank(i,B)$, for some $i \in [m]$, and let $l = 1$ and $h = n$.
Consider the variable $r$, initially set to $0$, and which at the end of the algorithm execution will store the value $\rank(i,B)$. 
We can compute $\select(j,B)$, where $j = l + \lfloor (h - l)/ 2 \rfloor$, in constant time.
If $\select(j,B) = i$, then $\rank(i,B) = j$, and consequently we stop the algorithm execution and we return $j$.
On the other hand, if $\select(j,B) < i$, then we know that $\rank(i,B) \geq j$, in this case we set $r = j$ and $l = j + 1$ and we repeat the operation.
Finally, if $\select(j,B) > i$, then $\rank(i,B) < j$ and we can set $h = j - 1$ for the next iteration of the binary search.
The algorithm execution eventually stops when $l > h$ holds, and in this case we return the value $r$.
Since at every iteration of this binary search we either return $\rank(i,B)$ or halve the number of ones taken into account, it follows that the next remark holds true.
\begin{remark}\label{remark:IDrank}
An ID representation of a bitvector $B$ of size $m$ and $n$ entries set to $1$ can support $\rank(i, B)$, for every $i \in [m]$, in $O(\log n)$ time.
\end{remark}
In this paper, we extensively use the ID representation proposed by Raman et al.~\cite[Theorem 4.6]{RRR}, which uses a number of bits upper-bounded by the formula $\log \binom{m}{n} + o(n) + O(\log \log m)$.
In addition, to represent $B$ and support full rank and select queries on it we use the data structure proposed by Pătrașcu~\cite[Theorem 2]{succincter}.
This solution takes at most $\log \binom{m}{n} + m /(\frac{\log m}{t})^t+ \tilde{O}(m^{3/4})$ bits of memory and supports the two operations in $O(t)$ time.
Consequently, if $t$ is a constant with $t > 0$, then this data structure becomes a FID representation for the input bitvector.

\begin{lemma}\cite{RRR, succincter}\label{thm:RRR}
Given an arbitrary bitvector $B$ of size $m$ with $n$ ones, there exist:
\begin{enumerate}
    \item\label{thm:RRR property 2} an ID representation of $B$ taking at most $\log \binom{m}{n} + o(n) + O(\log \log m)$ bits.
    \item\label{thm:RRR property 1} a FID representation of $B$ taking at most $\log \binom{m}{n} + m / (\frac{\log m}{t})^t + \tilde{O}(m^{3/4})$ bits, where $t$ is any constant satisfying $t > 0$.
\end{enumerate}
\end{lemma}

\section{On the number of tries with a given symbol distribution}\label{section counting tries}

This section is dedicated to prove the following result.

\begin{theorem}\label{theorem number of tries}
Let $\mathcal{U}$ be the set of tries with $n$ nodes and labels drawn from an alphabet $\Sigma$, where each symbol $c \in \Sigma$ labels $n_c$ edges, then $|\mathcal{U}| = \frac{1}{n}\prod_{c \in \Sigma} \binom{n}{n_c}$
\end{theorem}

The same problem has already been addressed by Prezza (see Section~3.1 of~\cite{rindexTrie}) and Prodinger~\cite{TRProdinger}.
Prezza claimed that $\lvert \mathcal{U} \rvert$ is equal to $\prod_{c \in \Sigma}\binom{n}{n_c}$, however, this formula overestimates the correct number of tries by a multiplicative factor $n$.
On the other hand, Prodinger obtained the correct number in a technical report using the Lagrange Inversion Theorem~\cite{analyticCombinatorics}.
In this paper, we give an alternative proof using a simple bijection between these tries and a class of binary matrices. 
For the more general problem of counting tries with $n$ nodes and $\sigma$ characters, the corresponding formula $\mathcal{S}(n,\sigma)$ is well-known and in particular it is $\mathcal{S}(n,\sigma)=\frac{1}{n}\binom{n\sigma}{n-1}$~\cite{concreteMathematics, cycleLemma}. 
We now introduce the definition of set $\mathcal{M}$, which depends on $\mathcal{U}$.

\begin{definition}[Set $\mathcal{M}$]\label{Set M}
    We define $\mathcal{M}$ as the set of all $\sigma \times n$ binary matrices $M$, such that, for every $i \in [\sigma]$, it holds that $\ones(M[i][-]) = n_c$, with $c$ being the $i$-th character of $\Sigma$ according to its total order $\preceq$.
\end{definition}

Note that $\ones(M) = n - 1$, trivially follows from $\sum_{c \in \Sigma} n_c = n-1$.
In order to count the elements in the set $\mathcal U$, we now define a function $f$ mapping each trie $\mathcal{T}$ into an element of $M \in \mathcal{M}$.
According to our function $f$, the number of ones in the columns of $M$ encodes the topology of the trie, while the fact that $M$ is binary encodes the standard trie labeling constraint. 

\begin{definition}[Function $f$]
    We define the function $f : \mathcal{U} \rightarrow \mathcal{M}$ as follows: given a trie $\mathcal{T} \in \mathcal U$ consider its nodes $u_{1}, \ldots, u_{n}$ in pre-order visit.
    Then $M = f(\mathcal{T})$ is the $\sigma \times n$ binary matrix such that $M[i][j] = 1$ if and only if $c \in out(u_{j})$, where $c$ is the $i$-th character of $\Sigma$ according to $\preceq$.
\end{definition}

Observe that, given a trie $\mathcal{T}$, the $i$-th column $M[-][i]$ of the corresponding matrix $M=f(\mathcal{T})$ is the characteristic bitvector of the outgoing labels of node $u_i$ (under the specific ordering on $\Sigma$) and in particular it holds that $\ones(M[-][i])$ is the out-degree of $u_i$.
We observe that the function $f$ is injective (see Remark~\ref{remark: injective}), however, $f$ is not surjective in general: some matrices $M \in \mathcal{M}$ may not belong to the image of $f$ because during the inversion process, connectivity constraints could be violated (see Figure~\ref{fig1}).
To characterize the image of $f$ and to prove its injectivity, we define the following two sequences.

\begin{definition}[Sequences $\DD$ and $\LL$]\label{arrays D and L}
    For every $M \in \mathcal{M}$ we define the integer sequences $\DD$ and $\LL$ of length $n$, such that $\DD[i] = \ones(M[-][i]) - 1$ and $\LL[i] = \sum_{j = 1}^{i} \DD[j]$ for every $i \in [n]$.
\end{definition}

We can observe that if $M \in f(\mathcal{U})$, then if we consider the unique trie $\mathcal{T}=f^{-1}(M)$, $\DD[i]$ is the out-degree of the $i$-th node in pre-order visit minus one, namely $\DD[i] = \rvert out(u_{i})\lvert - 1$.
Consequently, we have that $\LL[i] = \sum_{j = 1}^{i}\lvert out(u_j) \rvert - i$, and therefore $\LL[i] + 1$ denotes the total number of ``pending'' edges immediately after the pre-order visit of the node $u_i$, where the $+1$ term stems from the fact that the root has no incoming edge.
In other words, $\LL[i] + 1$ denotes the number of edges whose source has already been visited, while their destination has not.
We note that for every $M \in \mathcal{M}$, it holds that $\LL[n] = -1$, since $\ones(M) = n - 1$.
In the following, we show that the sequence $\LL$ can be used to determine if $M \in f(\mathcal{U})$ holds for a given matrix $M \in \mathcal{M}$.
In fact, we will see that if $M \in f(\mathcal{U})$, the array $\LL$ is a so-called \emph{Łukasiewicz path} as defined in the book of Flajolet and Sedgewick~\cite{analyticCombinatorics} (see Section~I. 5. ``Tree Structures'').
Specifically, a Łukasiewicz path $\mathcal{L}$ is a sequence of $n$ integers satisfying the following conditions.
(i)~$\mathcal{L}[n] = -1$, and for every $i \in [n-1]$, it holds that (ii)~$\mathcal{L}[i]\geq 0$, and (iii)~$\mathcal{L}[i+1] - \mathcal{L}[i] \geq -1$.
Łukasiewicz paths can be used to encode the topology of a trie since it is known that their are in bijection with unlabeled ordered trees of $n$ nodes~\cite{analyticCombinatorics}.
Specifically, given an ordered unlabeled tree $\mathcal{T}'$ the $i$-th point in its corresponding Łukasiewicz path $\mathcal{L}$ is $\mathcal{L}[i]= \sum_{j = 1}^{i} (\lvert out(u_j)\rvert -1)$, i.e., $\mathcal{L}$ is obtained by prefix-summing the sequence formed by the nodes out-degree minus $1$ in pre-order.
The reverse process consists in scanning $\mathcal{L}$ left-to-right and appending the current node $u_i$ of out-degree $\mathcal{L}[i]-\mathcal{L}[i-1]+1$ to the deepest pending edge on the left; obviously $u_1$ is the root of $\mathcal{T}'$ where we assume $\mathcal{L}[0]=0$.

By using the notion of Łukasiewicz paths is easy to see that $f$ is injective. Consider two distinct tries $\mathcal{T}$ and $\mathcal{T'}$ and their matrices $M = f(\mathcal T)$ and $M' = f(\mathcal{T}')$. If the trie topologies differ, the sequences $\LL$ and $\LL'$ corresponding to $M$ and $M'$ respectively are different Łukasiewicz paths. This implies that there exists $i$ such that $\DD[i]\neq \DD'[i]$ and therefore $M[-][i]\neq M'[-][i]$ since $\ones(M[-][i]) \neq \ones (M'[-][i])$.
Otherwise, $\mathcal{T}$ and $\mathcal{T'}$ differ by an outgoing label of their respective $i$-th node in pre-order, for some $i\in [n]$, and consequently again $M[-][i]\neq M'[-][i]$ necessarily holds due to the different edge labels.

\begin{remark}\label{remark: injective}
The function $f$ is injective.
\end{remark}

The next result states that the tries with a fixed number of occurrences of symbols are in bijection with the matrices of $\mathcal{M}$ whose corresponding array $\LL$ is a \emph{Łukasiewicz path}.
Similar results have been proved in other articles~\cite{aProblemOfArrangements, cycleLemma, roteCountingDegrees}.

\begin{restatable}[]{lemma}{lemmai}\label{characterizing f}
A binary matrix $M \in \mathcal{M}$ belongs to $f(\mathcal{U})$ if and only if, the corresponding array $\LL$ is a Łukasiewicz path.
\end{restatable}

\begin{proof}
$(\Rightarrow):$ Consider the trie $\mathcal{T} = (V,E) = f^{-1}(M)$ and its underlying unlabeled ordered tree $\mathcal{T}'$, i.e., $\mathcal{T}'$ is the ordered tree obtained by deleting the labels from $\mathcal{T}$.
Obviously, the corresponding nodes of $\mathcal{T}$ and $\mathcal{T}'$ in pre-order have the same out-degree. Therefore, by Definition~\ref{arrays D and L}, the array $\LL$ of $M$ is equal to the Łukasiewicz path of $\mathcal{T}'$.
 
$(\Leftarrow):$ 
Given $M\in \mathcal{M}$ and its corresponding array $\LL$, we know by hypothesis that $\LL$ is a Łukasiewicz path. Therefore, we consider the unlabeled ordered tree $\mathcal{T}'$ obtained by inverting $\LL$.
Let $u_i$ be the $i$-th node of $\mathcal{T}'$ in pre-order, we label the $j$-th left-to-right edge outgoing from $u_i$ with the symbol corresponding to the $j$-th top-down $1$ in $M[-][i]$. The resulting order labeled tree is a trie $\mathcal{T}''$ and by construction it follows that $f(\mathcal{T}'')=M$.
\end{proof}

To count the number of tries in $\mathcal{U}$, we now introduce the concept of \emph{rotations}.

\begin{definition}[Rotations]\label{def rotations}
    For every matrix $M \in \mathcal{M}$ and integer $r \geq 0$, we define the $r$-th rotation of $M$, denoted by $M^{r}$, as $M^{r}[-][j] = M[-][(j+r-1)\bmod{n} +1]$ for every $j \in [n]$.
\end{definition}

Informally, a rotation $M^{r}$ of a matrix $M \in \mathcal{M}$ is obtained by moving the last $r$ columns of $M$ at its beginning, and therefore $M^0=M$. Moreover, since $M^r = M^{r\bmod {n}}$, in the following we consider only rotations $M^r$ with $r\in [0,n-1]$.
Note also that, for every $r \geq 0$, we know that $M^r \in \mathcal{M}$ as rotations do not change the number of entries equal to $1$ in a row. 
Moreover, we can observe that the array $\DD$ of a rotation $M^r$ corresponds to a cyclic permutation of the array $\DD'$ of $M$.
Next we show that all $n$ rotations of a matrix $M\in \mathcal{M}$ are distinct and exactly one rotation has an array $\LL$ which is Łukasiewicz.
To do this, we recall a known result from~\cite{roteCountingDegrees} and~\cite[Section~I. 5. ``Tree Structures'', Note I.47]{analyticCombinatorics}.

\begin{lemma}\cite[Lemma 2]{roteCountingDegrees}\label{rotation principle}
    Consider a sequence of integers $A = a_1, \ldots, a_n$ with $\sum_{i=1}^{n}a_i = -1$. Then (i)~all $n$ cyclic permutations of $A$ are distinct, and (ii)~there exists a unique cyclic permutation $A' = a_1', \ldots , a_n'$ of $A$ such that $\sum_{i=1}^{j}a_i' \geq 0$, for every $j \in [n-1]$
\end{lemma}

The following corollary is a consequence of the above lemma and an example is shown in Figure~\ref{fig1}.

\begin{restatable}[]{corollary}{corollaryi}\label{unique matrix in image}
    For every matrix $M \in \mathcal{M}$ all its rotations are distinct and only one of them belongs to $f(\mathcal{U)}$.
\end{restatable}

\begin{proof}
    Consider the sequence $\DD$ of matrix $M$. 
    We know that the sequence of integers $\DD[1], \ldots , \DD[n]$ sums to $-1$.
    Consequently, by Lemma~\ref{rotation principle}, it follows that all the cyclic permutations of $\DD$ are distinct and thus all the corresponding rotations of $M$ are distinct too.
    By Lemma~\ref{characterizing f}, to conclude the proof it remains to be shown that there exists a unique rotation $M^r$ whose array $\LL$ is a Łukasiewicz path.
    By Lemma~\ref{rotation principle}, we know there exists a unique rotation $M^r$ for which its corresponding array $\LL$ is nonnegative excluding the last position, which is a necessary condition for $\LL$ to be a Łukasiewicz path.
    Therefore $M^r$ is the only candidate that can belong to $f(\mathcal{U})$.
    Moreover, we know that $\LL[n] =-1$ and by Definition~\ref{arrays D and L} it holds that $\LL[i+1]-\LL[i]=\ones(M^r[-][i])-1 \geq -1$ for every $i \in [n-1]$. 
    Since these observations prove that $\LL$ is Łukasiewicz, it follows that $M^r \in f(\mathcal{U})$.
\end{proof}

\makeatletter
\newcommand\notsotiny{\@setfontsize\notsotiny\@vipt\@viipt}
\makeatother

\begin{figure}[h]
	\centering
	\renewcommand{\arraystretch}{0.85}
	\begin{subfigure}{0.5\textwidth}
    \vspace{0pt}
	\begin{tabular}{ P{1.5em} | P{1.5em} P{1.5em} P{1.5em} P{1.5em} P{1.5em} P{1.5em} P{1.5em} |}
		   & \scriptsize{1} & \scriptsize{2} & \scriptsize{3} & \scriptsize{4} & \scriptsize{5} & \scriptsize{6} & \scriptsize{7} \\
		 \hline
		 \footnotesize{a} & 0 & 1 & 0 & 0 & 0 & 0 & 0 \\
		 \footnotesize{b} & 1 & 0 & 0 & 0 & 1 & 1 & 0 \\
		 \footnotesize{c} & 0 & 1 & 0 & 0 & 0 & 1 & 0 \\
		 \hline

	\end{tabular}
		
	\begin{tabular}{ P{1.5em} P{1.5em} P{1.5em} P{1.5em} P{1.5em} P{1.5em} P{1.5em} P{1.5em} }

		  $\DD$ & 0 & 1 & -1 & -1 & 0 & 1 & -1 \\
		$\LL$ & 0 & 1 & 0 & -1 & -1 & 0 & -1 \\

	\end{tabular}

	\vspace{1.4em}
	
	\begin{tabular}{ P{1.5em} | P{1.5em} P{1.5em} P{1.5em} P{1.5em} P{1.5em} P{1.5em} P{1.5em} |}
	& \scriptsize{1} & \scriptsize{2} & \scriptsize{3} & \scriptsize{4} & \scriptsize{5} & \scriptsize{6} & \scriptsize{7} \\
	\hline
	\footnotesize{a} & 0 & 0 & 0 & 0 & 1 & 0 & 0 \\
	\footnotesize{b} & 1 & 1 & 0 & 1 & 0 & 0 & 0 \\
	\footnotesize{c} & 0 & 1 & 0 & 0 & 1 & 0 & 0 \\
	\hline
	
	\end{tabular}
		
	\begin{tabular}{ P{1.5em} P{1.5em} P{1.5em} P{1.5em} P{1.5em} P{1.5em} P{1.5em} P{1.5em} }
		
		$\DD$ & 0 & 1 & -1 & 0 & 1 & -1 & -1 \\
		$\LL$ & 0 & 1 & 0 & 0 & 1 & 0 & -1 \\
		
	\end{tabular}

	\end{subfigure}
	\hspace{0.025\textwidth}
	\begin{subfigure}{0.4\textwidth}
	\centering
    \raisebox{-0.8em}{
	\begin{tikzpicture}[
		dim/.style={minimum size=1em, font={\notsotiny}}, 
		scale = 0.7,
		ghost/.style={draw=none},
        writes/.style={font={\small}},
		dots/.style={text centered, font={\small}},
        >={Stealth[length=1.4mm, width=1.2mm]}]

		\node[state, dim] (1) at (0,0) {1};
		\node[state, dim] (2) at (0,-1.2) {2};
		\node[state, dim] (3) at (-0.7,-2.4) {3};
		\node[state, dim] (4) at (0.7,-2.4) {4};
		\node[state, dim] (5) at (0,-3.6) {5};
		\node[state, dim] (6) at (0,-4.8) {6};
		\node[state, dim] (7) at (-0.7,-6) {7};
		\node[state, dim, ghost] (8) at (0.7,-6) {};

        \node[dots] (9) at (0, -7.25) {(a)};
		
		\draw[->] (1) to node [right, writes] {$b$} (2);
		\draw[->] (2) to node [right, writes] {$a$} (3);
		\draw[->] (2) to node [right, writes] {$c$} (4);
		\draw[->] (5) to node [right, writes] {$b$} (6);
		\draw[->] (6) to node [right, writes] {$b$} (7);
		\draw[->] (6) to node [right, writes] {$c$} (8);

        \begin{scope}[xshift=3.5cm]

		\node[state, dim] (1) at (0,0) {1};
		\node[state, dim] (2) at (0,-1.2) {2};
		\node[state, dim] (3) at (-0.7,-2.4) {3};
		\node[state, dim] (4) at (0.7,-2.4) {4};
		\node[state, dim] (5) at (0.7,-3.6) {5};
		\node[state, dim] (6) at (0,-4.8) {6};
		\node[state, dim] (7) at (1.4,-4.8) {7};
			
        \node[dots] (9) at (0, -7.25) {(b)};
            
		\draw[->] (1) to node [right, writes] {$b$} (2);
		\draw[->] (2) to node [right, writes] {$b$} (3);
		\draw[->] (2) to node [right, writes] {$c$} (4);
		\draw[->] (4) to node [right, writes] {$b$} (5);
		\draw[->] (5) to node [right, writes] {$a$} (6);
		\draw[->] (5) to node [right, writes] {$c$} (7);
        \end{scope}
		\end{tikzpicture}
        }
	\end{subfigure}

    \caption{In the top-left corner, the figure shows a $3\times7$ binary matrix $M$ with $n_a = 1$, $n_b = 3$, and $n_c = 2$, and its sequences $\DD$ and $\LL$ (here we omit the special character $\#$). 
    Since $\LL$ is not Łukasiewicz, by inverting $M$, we obtain the object~(a) which is not a trie.
    This happens because there are no pending edges to which the pre-order node $5$ can be attached since $\LL[4] = -1$.
    However, the rotated matrix $M^3$ showed in the bottom-left corner produce the valid trie shown in~(b).
    In this case $\LL$ is Łukasiewicz.
    }
    \label{fig1}
\end{figure}
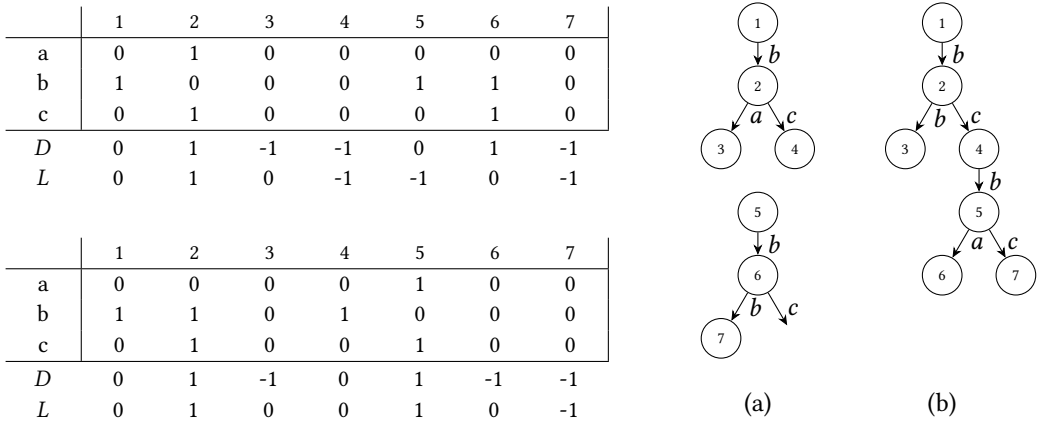

Moreover, it is known~\cite[Section~I. 5. ``Tree Structures'', Note I.47]{analyticCombinatorics} that given a Łukasiewicz path $\LL$ ending in $\LL[n]=-1$ and assuming $\LL[0] = 0$, allowed to assume negative values in the intermediate positions, we can make it a proper Łukasiewicz path by rotating its corresponding sequence $\DD$ by $n-i$ columns, where $i$ is the position of the leftmost minimum in $\LL$. Therefore the unique rotation $M^r$ belonging to $f(\mathcal{U)}$, mentioned in the above proof, is the one corresponding to $r=n-i$.
In the following, we prove the main result of this section.

\begin{proof}[Proof of Theorem~\ref{theorem number of tries}]
    Let $\sim$ be the relation over $\mathcal{M}$ defined as follows.
    For every $M, \bar M \in \mathcal{M}$ we have that $M \sim \bar M$ holds if and only if, there exists an integer $i$ such that $M^i = \bar M$.
    It is easy to show that $\sim$ is an equivalence relation over $\mathcal{M}$.
    Indeed, $\sim$ satisfies; (i)~reflexivity ($M \sim M$ since $M^0 = M$), (ii)~symmetry (if $M^i = \bar M$, then $\bar M^{(n-i)} = M$), and (iii)~transitivity (if $M^i = \bar M$ and $\bar M^j = \hat M$, then $M^{i+j} = \hat M$).
    By Corollary~\ref{unique matrix in image}, we know that for every $M \in \mathcal{M}$ the equivalence class $[M]_{\sim}$ contains $n$ elements with a unique matrix $\bar{M} \in [M]_{\sim}$ such that $\bar M \in f(\mathcal{U})$.
    Therefore, since by Remark~\ref{remark: injective} $f$ is injective, we deduce that $\lvert \mathcal{U} \rvert = \lvert f(\mathcal{U}) \rvert = \lvert \mathcal{M} \rvert / n$.
    Finally, since it is easy to observe that $\lvert \mathcal{M} \rvert = \prod_{c \in \Sigma}\binom{n}{n_c}$ it follows that $\lvert \mathcal{U} \rvert = \frac{1}{n}\prod_{c \in \Sigma}\binom{n}{n_c}$.
\end{proof}

Finally, we note that by summing the quantities $\frac{1}{n}\prod_{c \in \Sigma} \binom{n}{n_c}$ over all the possible symbol distributions $\{n_c \mid c \in \Sigma \}$, satisfying the tree constraint $\sum_{c\in \Sigma} n_c= n - 1$, due to the generalized Vandermonde's convolution~\cite[Exercise 5.62]{concreteMathematics}, we obtain the well-known formula $\mathcal{S}(n,\sigma) = \frac{1}{n}\binom{n\sigma}{n-1}$~\cite{concreteMathematics, cycleLemma} representing the total number of tries with $n$ nodes over an alphabet of size $\sigma$.

\section{Entropy of a trie}\label{sec:entropy}

In this section we define a worst-case entropy and an empirical entropy for tries taking into account the distribution of the edge labels.
Later, we show the reachability of our empirical measure and we compare it with the label entropy of Ferragina et al.~\cite{xBWTconf}.
Our worst-case formula directly derives from the number of tries $\lvert \mathcal U \rvert $ considered in the previous section.

\begin{definition}[Worst-case entropy]\label{def: worst-case}
    The worst-case entropy $\mathcal{H}^{\wc}$ of a trie $\mathcal{T}$ with symbol distribution $\{n_c \mid c \in \Sigma \}$ is defined as;
    
    \[\mathcal{H}^{\wc}(\mathcal{T})  = \log_2 \lvert \mathcal{U} \rvert =  \sum_{c \in \Sigma} \log \binom{n}{n_c}  - \log n\]
\end{definition}

Clearly the above worst-case entropy is a lower-bound for the minimum number of bits (in the worst-case) required to encode a trie with a given number of symbol occurrences.
Note that $\mathcal{H}^{\wc}(\mathcal{T})$ is always smaller than or equal to the worst-case entropy $\mathcal{C}(n,\sigma) = \log \frac{1}{n} \binom{n\sigma}{n - 1}$~\cite{concreteMathematics, RRR} of tries having $n$ nodes over an alphabet of size $\sigma$ and it is much smaller when the symbol distribution is skewed.
Next we introduce the notion of \emph{$0$-th order empirical entropy} of a trie $\mathcal{T}$. 

\begin{definition}[$0$-th order empirical entropy]\label{def: 0-entropy}
    The $0$-th order empirical entropy $\mathcal{H}_0$ of a trie $\mathcal{T}$ with symbol distribution $\{n_c \mid c \in \Sigma \}$ is defined as;

    \[
    \mathcal{H}_{0}(\mathcal{T}) = \sum_{c \in \Sigma}\frac{n_c}{n}\log\left(\frac{n}{n_c}\right) + \frac{n - n_c}{n}\log\left( \frac{n}{n-n_c} \right)
    \]
\end{definition}

Given a bitvector $B_{n,m}$ of size $n$ with $m$ ones, the worst-case entropy is defined as $\mathcal{H}^{\wc}(B_{n,m}) = \log \binom{n}{m}$ while the corresponding $0$-th order empirical entropy is $\mathcal{H}_0(B_{n,m}) = \frac{m}{n}\log(\frac{n}{m}) + \frac{n-m}{n}\log(\frac{n}{n-m})$~\cite{compactDataStructures, elementsOfInformationTheory}.
Furthermore, the following known inequalities relate the above entropies of a binary sequence: it holds that $n\mathcal{H}_0(B_{n,m})-\log(n+1) \leq \mathcal{H}^{\wc}(B_{n,m})\leq n\mathcal{H}_0(B_{n,m})$~\cite[Equation 11.40]{elementsOfInformationTheory}.
As a consequence of these inequalities, by summing the quantities $\log\binom{n}{n_c}$ of Definition~\ref{def: worst-case}, for every possible $c \in \Sigma$, we directly obtain the following result.

\begin{lemma}\label{lemma:emp_wc}
    For every trie $\mathcal{T}$ with $n$ nodes, the following relations hold $n\mathcal{H}_0(\mathcal{T})-\sigma\log(n+1)-\log n \leq \mathcal{H}^{\wc}(\mathcal{T})\leq n\mathcal{H}_0(\mathcal{T})- \log n$
\end{lemma}

Consequently, we can observe that the following corollary holds.

\begin{corollary}\label{corollary: relation worst-case empirical}
For every trie $\mathcal{T}$ it holds that
$n\mathcal{H}_0(\mathcal{T}) = \mathcal{H}^{\wc}(\mathcal T) + O(\sigma \log n)$
\end{corollary}

Note that the results of Lemma~\ref{lemma:emp_wc} and Corollary~\ref{corollary: relation worst-case empirical} are valid also if we consider the effective alphabet formed by all symbols labeling at least one edge in $\mathcal{T}$, as the other characters do not affect $\mathcal{H}^{\wc}$ or $n\mathcal{H}_0$.
The result of Corollary~\ref{corollary: relation worst-case empirical} is consistent with known results for strings: given a string $S$ of length $n$ over the alphabet $\Sigma$ it is known that $n\mathcal{H}_0(S)=\mathcal{H}^{\wc}(S)+O(\sigma\log n)$ (see~\cite[Section 2.3.2]{compactDataStructures}) where $\mathcal{H}^{\wc}(S)=\log\binom{n}{n_1,\ldots,n_\sigma}$ is the worst-case entropy for the set of strings with character occurrences $n_i$ and ${H}_0(S)$ is the zero-th order empirical entropy of $S$ as defined in Section~\ref{sec: Notation}.
We now aim to extend the $0$-th order trie entropy to higher orders.
To this purpose, similarly as in~\cite{xBWTconf, xBWTjournal, rindexTrie}, for every integer $k \geq 0$, we define the \emph{$k$-th order context} of  a node $u$, denoted by $\lambda_k(u)$, as the string formed by the last $k$ symbols of the path from the root to $u$.
Formally, we have $\lambda_0(u) = \epsilon$ and $\lambda_k(u) = \lambda_{k-1}(\pi(u)) \cdot \lambda(u)$.
By definition of $\pi$ (see Section~\ref{sec: Notation}), if a node $u$ has a depth $d$ with $d < k$, then such string is left-padded by the string $\#^{k-d}$.
Moreover, we introduce the integers $n_w$ and $n_{w,c}$. Informally, given $w \in \Sigma^k$, for some $k \geq 0$, then $n_w$ is the number of nodes in $\mathcal T$ having $k$-length context $w$, while $n_{w,c}$ is the number of nodes having $w$ as their $k$-th order context and an outgoing edge labeled by $c$.

\begin{definition}[integers $n_w$ and $n_{w,c}$]\label{def: integers nw and nwc}
    For every string $w \in \Sigma^*$ and character $c \in \Sigma$, we define the integers $n_w$ and $n_{w,c}$ as follows.
    \begin{itemize}
        \item $n_w = \lvert \{ u \in V : w = \lambda_k(u) \} \rvert $
        \item $n_{w,c} = \lvert \{ u \in V : w = \lambda_k(u) \text{ and } c \in out(u)\}\rvert$
    \end{itemize}
\end{definition}

With the next definition we generalize our notion of empirical entropy of a trie to higher orders in a similar way of what has been done in~\cite{xBWTjournal, xBWTconf} for labeled trees.

\begin{definition}[$k$-th order empirical entropy]\label{def: k-entropy}
    For every integer $k \geq 0$, the $k$-th order empirical entropy $\mathcal H_k(\mathcal T)$ of a trie $\mathcal{T}$ is defined as follows\footnote{We observe that a related $k$-th order \emph{worst-case} entropy formula, defined as $\sum_{w \in \Sigma^{k}} \sum_{c \in \Sigma} \log \binom{n_{w}}{n_{w,c}}$, appears in the article of Prezza~\cite[Section 3.1]{rindexTrie}. However, as shown in Section~\ref{section counting tries}, already for $k=0$ this formula does not correspond to the actual number of tries with a given symbol distribution. Hence, we propose a modified version of this entropy notion.};

    \[
    \mathcal{H}_{k}(\mathcal{T}) = \sum_{w \in  \Sigma^k}\sum_{c \in \Sigma}\frac{n_{w,c}}{n}\log\left(\frac{n_w}{n_{w,c}}\right) + \frac{n_w - n_{w,c}}{n}\log\left( \frac{n_w}{n_w-n_{w,c}} \right)
    \]
\end{definition}

Note that for $k = 0$ this formula coincides with that of Definition~\ref{def: 0-entropy} since in this case $\Sigma^0= \{\epsilon\}$ and $n_{\epsilon}=n$ and for every character $c \in \Sigma$, we have $n_{\epsilon,c} = n_c$.
Moreover, analogously to strings, by the log-sum inequality (see~\cite[Eq. (2.99)]{elementsOfInformationTheory}), for every trie $\mathcal{T}$ and integer $k\geq 0$ it holds that $H_{k+1}(\mathcal{T}) \leq H_{k}(\mathcal{T})$.

\subsection{Entropy reachability}

In this subsection we show that, given any non-negative integer $k$, we can compress an arbitrary trie $\mathcal T$ within $n\mathcal{H}_k(\mathcal{T})+2$ bits of space, plus an additional term required to encode the integers $n_w$ and $n_{w,c}$.

\begin{theorem}\label{theorem: reachability}
    Let $\mathcal{T}$ be a trie and $k$ an integer with $k \geq 0$.
    Then it is possible to represent $\mathcal{T}$ in a number of bits upper-bounded by $n\mathcal{H}_k(\mathcal{T}) + 2 + (\sigma+1)\sigma^{k} \lceil\log n\rceil$.
\end{theorem}

To prove this theorem, we propose an extended version of the arithmetic coding algorithm for tries.
Arithmetic coding is a lossless compression algorithm introduced by Elias, which compresses a string $S$ of length $n$ to its $0$-th order (string) entropy $n\mathcal{H}_0(S)$~\cite[Section 12.2]{magicAlgorithms}.
Given the probabilities $p_c$ associated with every symbols $c$ in $S$, the original algorithm for strings computes in $n$ iterations a real-valued interval $[l,l+s)\subseteq [0,1)$ corresponding to $S$.
The algorithms works as follows: the initial interval is $[0,1)$, then during the $i$-th iteration it shrinks the current interval $[l_{i-1},l_{i-1}+s_{i-1})$ corresponding to the prefix $S[1..i-1]$ of $S$ depending on the probability $p_{S[i]}$ associated to the current character $S[i]$.
In particular, the new updated interval is $[l_i,l_i+s_i)$, where $s_i = s_{i-1}p_{S[i]}$ and $l_i = l_{i-1} + f_{S[i]} s_{i-1}$, with $f_{S[i]}$ defined as the sum of the probabilities associated to the characters strictly smaller than $S[i]$.
The final interval produced during the last iteration is $[l,l+s)$ with $l=l_n$ and $s=s_n$.
The probability distribution of the symbols, together with any number $x$ contained in such an interval $[l,l+s)$ uniquely identifies the string $S$ and thus can be used to encode $S$. 
Specifically, arithmetic coding chooses as value $x$ the middle point $\middlealg = l + \frac{1}{2}s$ of the final interval and emits the binary representation of $\middlealg$ truncated to its $d = \lceil \log \frac{2}{s}\rceil$ most significant bits.
Formally, let $b_1,b_2,b_3, \ldots$ be the infinite bit sequence representing $\middlealg$ in base $2$, that is $\middlealg = \sum_{i\geq1}b_i2^{-i}$.
The encoding algorithm emits the bit sequence $b_1,b_2,\ldots,b_d$ representing the truncated number $\sum_{i=1}^{d}b_i2^{-i}$.
As stated in the next lemma, this truncated number still belongs to the interval $[l,l+s)$, and thus the representation can be inverted. 

\begin{lemma}\cite[Corollary 12.1]{magicAlgorithms}\label{corollary: truncation}
    For every interval $[l,l+s)$ contained in $[0,1)$, the truncation of $\middlealg=l+\frac{s}{2}$ to its first $\lceil \log \frac{2}{s}\rceil$ bits still falls in the interval $[l,l+s)$.
\end{lemma}

The claimed space bound for arithmetic coding follows from the fact that the integer $d$ is at most $n\mathcal H_0(S) + 2$ where $\mathcal{H}_0(S)$ is the $0$-th order entropy of the string $S$~\cite[Theorem 12.3]{magicAlgorithms}. 
Given the encoder output $x$ and the probability distribution of the symbols, the decoder can reconstruct the original string $S$, from left to right iteratively as follows. 
The initial interval is $[0,1)$, and if the current interval is $[l_i,l_i+s_i)$ the decoder divides it in sub-intervals according to the cumulative distribution $f_c$ of the symbols and it determines $S[i]$ from the sub-interval including $x$.
The next interval $[l_{i+1},l_{i+1}+s_{i+1})$ is computed as done during the compression phase. 
For a more detailed description and analysis of arithmetic coding see the book of Ferragina~\cite[Section 12.2]{magicAlgorithms}.
Another known important feature of arithmetic coding is that it can achieve higher-order compression if during the encoding phase we consider the conditional probability that a symbol occurs knowing its $k$ preceding characters.
Formally, with arithmetic coding we can store a string $S$ of length $n$ within $n \mathcal{H}_{k}(S) + 2$ bits of space, assuming that all character probabilities are stored for every possible $k$-length context.
In the following, we show that this algorithm can be adapted to tries for the same purpose.

\paragraph{Compression algorithm.}\label{par:encoder} The following adaptation of arithmetic coding for tries works similarly to the original one for strings~\cite[Algorithm 12.2]{magicAlgorithms}. 
The algorithm executes $n$ iterations in total, one for each node in the trie and it processes the nodes according to their order in a pre-order traversal.
At the beginning, the initial interval is set to $[0, 1)$, then in the $i$-th iteration the current interval is reduced based on the following information: (i)~the context of the current node $u_i$, i.e., the string $\lambda_k(u_i)$ and (ii)~the probabilities associated to the set of edges outgoing from $u_i$ given the context $\lambda_k(u_i)$.
At the end of the algorithm execution, the resulting interval together with the symbol probabilities will uniquely encode the input trie $\mathcal{T}$, and to decompress $\mathcal{T}$ it will be sufficient to store a single number in that interval.
To explain the compression algorithm in detail, we now introduce some further notation.
Given an integer $i\in [n]$, a character $c \in \Sigma$, and a string $w \in \Sigma^k$, we define $p_w(c)$ as the fraction of nodes in $\mathcal{T}$, among those with context $w$, having an outgoing edge labeled $c$, formally $p_w(c)=\frac{n_{w,c}}{n_w}$, where $n_w$ and $n_{w,c}$ are the integers introduced in Definition~\ref{def: integers nw and nwc}.
Consider the nodes $u_1, u_2, \ldots, u_n$ of a trie $\mathcal{T}$ sorted according to a pre-order visit.
In the $i$-th iteration, that is during the encoding of $u_i$, we iterate over the $\sigma$ alphabet symbols $c_j$ in the alphabet order $\preceq$.
Let $I_{i,j-1} = [l_{i,j-1}, l_{i,j-1} + s_{i,j-1})$ be the interval we obtained immediately before processing node $u_i$ and character $c_j$, i.e., the $j$-th character according to $\preceq$.
Then in the next iteration we consider this symbol $c_j$ and, depending on whether $u_i$ has an outgoing edge labeled $c_j$, we update our interval.
Intuitively, the encoder partitions the current interval $I_{i,j-1}$ into $2$ sub-intervals: the upper interval $I_{i,j-1}^{c_j}$ and the lower interval $I_{i,j-1}^{\lnot c_j}$ corresponding to the two events $c_j\in out(u_i)$ and $c_j\notin out(u_i)$, respectively.
The size of these sub-intervals (with respect to $s_{i,j-1}$) is assigned proportionally to the probabilities $p_w(c_j)$ or $1-p_w(c_j)$ of the corresponding events.
Next, the encoder updates the current interval to the sub-interval corresponding to the actual outcome, that is depending on whether $c_j\in out(u_i)$ or not in $\mathcal{T}$ it sets $I_{i,j}=I_{i,j-1}^{c_j}$ or $I_{i,j}=I_{i,j-1}^{\lnot c_j}$.
Formally, If $c_j\in out(u_i)$ then we update the interval size as $s_{i,j} = s_{i,j-1}\, p_w(c_j)$ and the starting point of the interval becomes $l_{i,j} = l_{i,j-1} + s_{i,j-1}\, (1-p_w(c_j))$, otherwise the new interval size is $s_{i,j} = s_{i,j-1}\, (1-p_w(c_j))$ and the left extreme of the interval is unchanged, i.e., $l_{i,j}=l_{i,j-1}$.
Thus, after these operations, the interval is updated to $I_{i,j} = [l_{i,j},l_{i,j}+s_{i,j})$.
At the end of the algorithm execution, the final interval corresponding to $\mathcal{T}$ is $[l,l + s)$ where $l = l_{n,\sigma}$ and $s = s_{n,\sigma}$.

In the next paragraph, we will show that any number falling in $[l,l + s)$, together with the conditional probabilities, is sufficient to reconstruct the input trie $\mathcal{T}$, this includes both the labels and the topology of the trie.
As in the classical arithmetic encoding, we output the binary representation of the middle point $\middlealg = l + \frac{1}{2}s$, truncated to its first $d = \lceil \log \frac{2}{s}\rceil$ most significant digits.
By Lemma~\ref{corollary: truncation} this number still lies in the final interval.
Later, we will prove that the number $d$ of bits emitted by our algorithm is upper-bounded by $n\mathcal{H}_k(\mathcal{T})+2$.
Algorithm~\ref{arithmetic compression} shows the pseudocode of this encoding scheme and Figure~\ref{fig:arithmetic_coding} depicts a sample execution.

\begin{algorithm}[!ht]
\LinesNumbered
\DontPrintSemicolon
\caption{Arithmetic coding for tries: compression scheme}\label{arithmetic compression}
    \Input{A trie $\mathcal{T}$ with $n$ nodes, an integer $k$ with $k \geq 0$}
    \Output{A number $output$ of a subinterval $[l, l + s)$ of $[0,1)$}

    \medskip

    $l \gets 0$, $s \gets 1$\tcp*{initial interval [0,1)}

    \smallskip

    \For{$i \gets  1$ \textnormal{\textbf{to}} $n$}{

        $u_i \gets$ $i$-th node of $\mathcal{T}$ in pre-order traversal\;
        $w \gets \lambda_k(u_i)$\tcp*{$k$-length context of node $u_i$}

        \smallskip

        \For{$j \gets 1$ \textnormal{\textbf{to}} $\sigma$}{

            \smallskip

            $c \gets$ $j$-th character of $\Sigma$\;
            
            $p_w(c) \gets \frac{n_{w,c}}{n_w}$\;

            \smallskip
            
            \If{$c \in out(u_i)$}{\label{line:start updating int}
                $l \gets l + s*(1 - p_w(c))$\tcp*{updating the starting point}
                $s \gets s * p_w(c) $\tcp*{updating the interval size}
            }
            \lElse{$s \gets s*(1 - p_w(c))$\tcp*[f]{update only the interval size}}\label{line:end updating int}
        }
    }        

    $\middlealg \gets l + \frac{1}{2}s$\tcp*{middle point of the interval $[l, l + s)$}

    \smallskip

    $d \gets \lceil \log_2 \frac{2}{s} \rceil$\;

    \smallskip

    $b_1,b_2,b_3, \ldots  \gets $ binary representation of $\middlealg$\tcp*{$\middlealg = \sum_{i\geq1}b_i2^{-i}$}
    \Return the sequence $b_1,b_2, \ldots , b_d$\;

\end{algorithm}

\paragraph{Decompression algorithm.} Next we show how to reconstruct the original trie $\mathcal{T}$ from the conditional probabilities $\frac{n_{w,c}}{n_w}$ and $1-\frac{n_{w,c}}{n_w}$ together with the real number $\valnumber= b_1,\ldots,b_d$, returned by Algorithm~\ref{arithmetic compression}. 
The decompression algorithm processes the nodes $u_1,\ldots,u_n$ according to their ordering in a pre-order traversal of $\mathcal{T}$. 
During the decoding of $u_i$, the decoder adds the outgoing edges of $u_i$ to the current trie in the order specified by $\preceq$.
Similarly to the original string algorithm, the decoder determines whether $c\in out(u_i)$ holds or not by checking which of the two non-overlapping sub-intervals, corresponding to the two events, contains the real number $\valnumber$ and updates the interval accordingly as done by the encoder.
See Algorithm~\ref{arithmetic decompression} for a pseudocode of this algorithm.

Formally, the decoding algorithm starts with the interval $[0,1)$, a trie $\mathcal{T} = (V,E)$ with $V=\{u_1\}$ and $E = \emptyset$ and proceeds as follows.
Let $\mathcal{T}_{i,j-1}$ and $[l_{i,j-1},l_{i,j-1}+s_{i,j-1})$ be respectively the (partially) decompressed trie and the interval obtained during the decoding of $u_i$ immediately before processing the $j$-th symbol $c_j$ of $\Sigma$.
By Algorithm~\ref{arithmetic compression} it holds that $c_j \notin out(u_i)$ if and only if $\valnumber \in [l_{i,j-1},l_{i,j-1}+(1 - p_{w}(c))s_{i,j-1})$.
Therefore, the decoder adds a new node $v$ and a new edge $(u_i,v,c_j)$ to $T_{i,j}$ if and only if $c_j \in out(u_i)$.
Before considering the next symbol $c_{j+1}\in \Sigma$ the decoder updates the current interval depending on whether $c_j \in out(u_i)$ or not using the respective probabilities $p_w(c_j)$ or $1-p_w(c_j)$ in the same way as done in Algorithm~\ref{arithmetic compression}.
We note that since $T_{i,j}$ already includes all the nodes preceding $u_i$ in pre-order along with their outgoing edges, it is possible to compute the context $w=\lambda_k(u_i)$ and therefore the associated probabilities $p_w(c)$ and $1-p_w(c)$ for every $c\in \Sigma$.
The algorithm terminates after the decoding of $u_n$.


\begin{algorithm}[!ht]
\LinesNumbered
\DontPrintSemicolon
\caption{Arithmetic coding for tries: decompression scheme}\label{arithmetic decompression}
    \Input{A real number $\valnumber \in [0,1)$, an integer $k$ with $k \geq 0$, the integers $n_{w}$ and $n_{w,c}$ for every $w \in \Sigma^{k}$ and $c \in \Sigma$}
    \Output{The decompressed trie $\mathcal{T} = (V,E)$}

    \medskip

    $V \gets \{ u_1\}$, $E \gets \emptyset$\tcp*{$u_1$ is the root of $\mathcal T$}

    $l \gets 0$, $s \gets 1$\tcp*{initial interval $[0,1)$}

    \smallskip

    \For{$i \gets  1$ \textnormal{\textbf{to}} $n$}{ \label{line: outloop}
        $u_i \gets i$-th node in a pre-order traversal of $\mathcal{T}$\;
        $w \gets \lambda_k(u_i)$\tcp*{$k$-length context of node $u_i$}

        \smallskip
        
        $p_w(c) \gets \frac{n_{w,c}}{n_w}$\;

        \smallskip

        \For{$j \gets  1$ \textnormal{\textbf{to}} $\sigma$}{ \label{line: inloop}
            $c \gets j$-th character of $\Sigma$\;

            \smallskip
            \If{$\valnumber < l + s*(1 - p_w(c))$}{
                $s = s* (1 - p_w(c)) $\tcp*{$c$ does not belong to $out(u_i)$}
            }
            \Else{
                $s \gets s * p_w(c)$\tcp*{$c$ belongs to $out(u_i)$}
                $l \gets l + s * (1 - p_w(c))$\;

                \smallskip
                
                $v \gets $ new node, $V \gets V \cup
                \{v\}$\tcp*{inserting new node $v$}
                $e \gets (u_i,v,c)$, $E \gets E \cup \{e\}$\tcp*{inserting new edge $e$}
            }
        }
    }

    \smallskip

    \Return $\mathcal{T} \gets (V,E)$\tcp*{final decompressed trie}
\end{algorithm}

\paragraph{Space complexity.} Finally, to prove Theorem~\ref{theorem: reachability}, we just need to show that we can store both the binary sequence returned by Algorithm~\ref{arithmetic compression} and the conditional probabilities of the symbols in at most $n\mathcal H_k (\mathcal T) + 2 + (\sigma + 1)\sigma^{k}\lceil \log n \rceil$ bits of space.
We show this result in the following proof.

\begin{proof}[Proof of Theorem~\ref{theorem: reachability}]
    The number emitted by Algorithm~\ref{arithmetic compression} is a binary sequence $b_1,b_2, \ldots, b_d$ of length $d=\lceil \log_2 \frac{2}{s} \rceil$ bits where $s$ is the final interval size.
    We preliminary prove that that $d < n \mathcal H_k (\mathcal T) + 2$.
    Since $\lceil \log_2 \frac{2}{s} \rceil < \log_2 \frac{2}{s} + 1 = 2 - \log s$, it sufficient to show that $- \log s \leq n \mathcal H_k (\mathcal T)$.
    Let $u_i$ be the $i$-th pre-order node and $w$ its length $k$ context.
    By Lines~\ref{line:start updating int}-\ref{line:end updating int} of Algorithm~\ref{arithmetic compression}, when processing $u_i$ the encoder shrinks the previous interval size $s_{i-1,\sigma}$ by a factor $\prod_{c \in out(u_i)}p_w(c) \prod_{c \notin out(u_i)}(1-p_w(c))$. 
    Accounting for all nodes and since initially is $s_0=1$, the final interval size is $s=\prod_{i=1}^n \prod_{c \in out(u_i)}p_w(c) \prod_{c \notin out(u_i)}(1-p_w(c))$.
    By grouping the nodes according to their context $w$, and since there are exactly $n_{w,c}$ nodes $u_i$ having $w = \lambda_k(u_i)$ and $c \in out(u_i)$, and $n_{w}-n_{w,c}$ nodes with $w = \lambda_k(u_i)$ and $c \notin out(u_i)$, the previous product can be rewritten as $s = \prod_{w\in \Sigma^k}\prod_{c \in \Sigma} p_w(c)^{n_{w,c}}(1-p_w(c))^{n_w - n_{w,c}}$, considering only the contexts $w$, with $n_w > 0$.
    Thus, by using the fact that $p_w(c)=\frac{n_{w,c}}{n_w}$ and $1-p_w(c)=\frac{n_{w}- n_{w,c}}{n_{w}}$ we obtain that:
    
    \begin{equation*}
        \begin{aligned}
            -\log s &= -\log \prod_{w \in \Sigma^k}\prod_{c \in \Sigma}\left(\frac{n_{w,c}}{n_w}\right)^{n_{w,c}}\left(\frac{n_{w}- n_{w,c}}{n_{w}}\right)^{n_w - n_{w,c}} = \\
            & = \sum_{w \in \Sigma^{k}}\sum_{c \in \Sigma} n_{w,c}\log \left(\frac{n_{w}}{n_{w,c}} \right) + (n_{w} - n_{w,c})\log\left( \frac{n_w}{n_w - n_{w,c}}\right)\\
                & = n\mathcal H_k (\mathcal{T})
        \end{aligned}
    \end{equation*}
    On the other hand, since the conditional probabilities have the forms $1 - \frac{n_{w,c}}{n_w}$ and $\frac{n_{w,c}}{n_w}$, to store them it is sufficient to save all integers $n_{w}$ and $n_{w,c}$ for every $w \in \Sigma^{k}$ and $c \in \Sigma$.
    Since there are at most $\sigma^{k}$ possible contexts of length $k$, it follows that we can save all the integers $n_{w}$ using $\sigma^{k} \lceil \log n \rceil$ bits, and all the integers $n_{w,c}$ in $\sigma^{k+1}\lceil \log n\rceil$ bits.
    Consequently, we can store all the conditional probabilities in $(\sigma + 1)\sigma^k \lceil \log n \rceil$ bits.
    It follows that we can represent any arbitrary trie $\mathcal{T}$ in at most $n\mathcal H_k(\mathcal T ) + 2 + (\sigma + 1)\sigma^{k}\lceil \log n \rceil$ bits of space.
\end{proof}

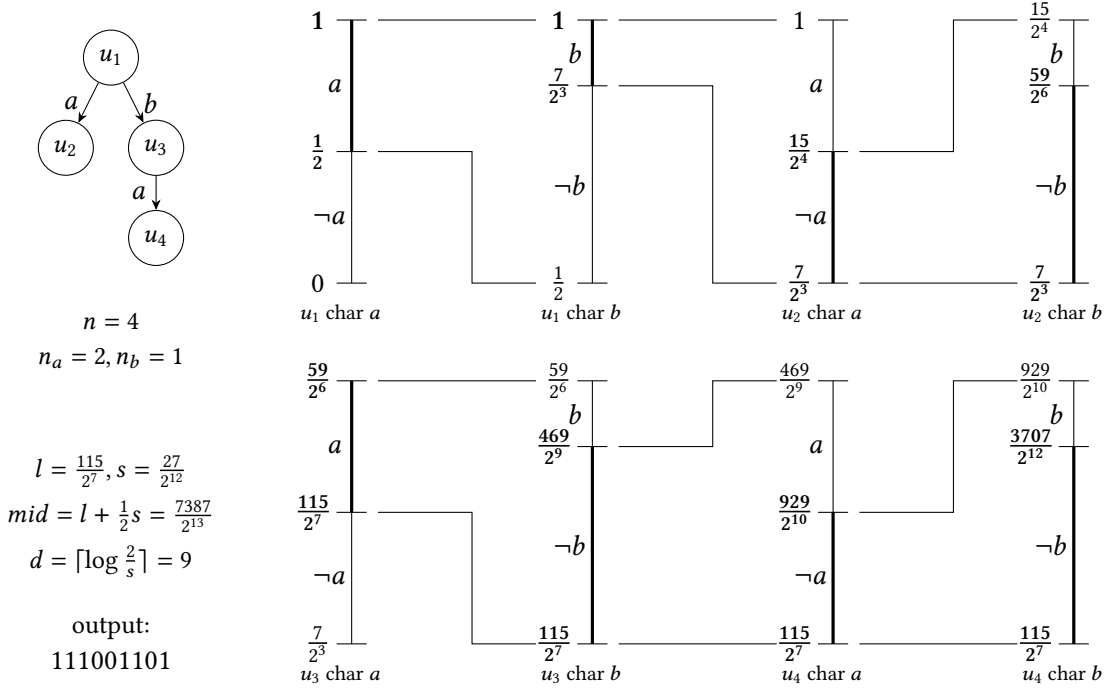
\begin{figure}[ht]\label{fig: encoding}
	\centering

	\begin{tikzpicture}[
        >={Stealth[length=1.4mm, width=1.2mm]},
        dim/.style={minimum size=1.9em, font={\small}},
    ]
        
        \begin{scope}[yshift = 3cm]
		\node[state, dim] (1) at (0,0) {$u_1$};
		\node[state, dim] (2) at (-0.6,-1.2) {$u_2$};
		\node[state, dim] (4) at (0.6,-2.4) {$u_4$};
		\node[state, dim] (3) at (0.6,-1.2) {$u_3$};

		\draw[->] (1) to node [left] {$a$} (2);
		\draw[->] (1) to node [right] {$b$} (3);
		\draw[->] (3) to node [left] {$a$} (4);

        \node (5) at (0,-3.5) {\small $n = 4$};
        \node (6) at (0,-4) {\small $n_a = 2, n_b = 1$};

        \node (8) at (0,-5.5) {\small $l = \frac{115}{2^7}, s= \frac{27}{2^{12}}$};
        \node (9) at (0,-6.1) {\small $\middlealg = l  + \frac{1}{2}s = \frac{7387}{2^{13}}$};
        
        \node (10) at (0,-6.7) {\small $d = \lceil \log \frac{2}{s} \rceil = 9$};

        \node[align = center] (11) at (0,-7.8) {\small output: \\$ 111001101$};

        \end{scope}

        \tikzmath{
            \a = 3.5; 
            \yspace = -4.8; 
            \xspace = 3.2; 
            \xnumspace = 0.45; 
            \aprob = 0.5;  
            \bprob = 0.75; 
            \xsepintexpstr = 0.35; 
            \xsepintexpend = 0.75; 
            \barxsize = 0.2; 
            \chardist = 0.3; 
            \subsdisty = 0.4; 
            \subsdistx = 0.15; 
            \subsbtwidst = 0.3; 
        }

        \begin{scope}[xshift = 3.2cm]
        \foreach \x/\y/\triple/\charchosen/\nodeid/\ign [count=\i] in {
            0/0/{0/0,1/1,1/0}/1/1/0, 
            1/0/{1/1,7/3,1/0}/1/1/0,
            2/0/{7/3,15/4,1/0}/0/2/0,
            3/0/{7/3,59/6,15/4}/0/2/0,
            0/1/{7/3,115/7,59/6}/1/3/0,
            1/1/{115/7,469/9,59/6}/0/3/0,
            2/1/{115/7,929/10,469/9}/0/4/0,
            3/1/{115/7,3707/12,929/10}/0/4/0
        } {
            
            \draw (\xspace*\x, \yspace*\y) to (\xspace*\x, \yspace*\y+ \a);

            \draw (-\barxsize + \xspace*\x, \yspace*\y) to (\barxsize + \xspace*\x, \yspace*\y);
            
            \draw (-\barxsize + \xspace*\x, \yspace*\y+\a) to (\barxsize + \xspace*\x, \yspace*\y+\a);

            \tikzmath{
                \prob = 0;
                \zoomindest = 0;
            }
            \ifthenelse{\isodd{\x}} {
                \tikzmath{\prob = \bprob;}
            } {
                \tikzmath{\prob = \aprob;}
            }

            \draw (-\barxsize + \xspace*\x, \yspace*\y+\prob*\a) to (\barxsize + \xspace*\x, \yspace*\y+\prob*\a);

            \ifthenelse{\isodd{\x}} {
                \node (\i) at (\xspace*\x - 0.75*\chardist, \yspace*\y + \a - 0.5*\a + 0.5*\prob*\a] ) {$b$};
                \node (-\i) at (\xspace*\x - \chardist, \yspace*\y + 0.5*\prob*\a) {$\neg b$};
            } {
                \node (\i) at (\xspace*\x - 0.75*\chardist, \yspace*\y + \a - 0.5*\a + 0.5*\prob*\a] ) {$a$};
                \node (-\i) at (\xspace*\x - \chardist, \yspace*\y + 0.5*\prob*\a) {$\neg a$};
            }

            \ifthenelse{\equal{\charchosen}{1}}{
                \tikzmath{\zoomindest = \yspace*\y;}
                \draw[very thick] (\xspace*\x, \yspace*\y+ \prob*\a) to (\xspace*\x, \yspace*\y+\a);

                \ifthenelse{\not\equal{\x}{3}}{
                    \draw (\xspace*\x + \xsepintexpstr, \yspace*\y+\a) to (\xspace*\x+\xspace-\xsepintexpend, \yspace*\y+\a);
                }
            } {
                \tikzmath{\zoomindest = \yspace*\y + \a;}
                \draw[very thick] (\xspace*\x, \yspace*\y) to (\xspace*\x, \yspace*\y+\prob*\a);

                \ifthenelse{\not\equal{\x}{3}}{
                    \draw (\xspace*\x + \xsepintexpstr, \yspace*\y) to (\xspace*\x+\xspace-\xsepintexpend, \yspace*\y);
                }
            }

            \ifthenelse{\not\equal{\x}{3}}{
                \draw (\xspace*\x + \xsepintexpstr, \yspace*\y+\prob*\a) -- (\xspace*\x + 0.5*\xspace, \yspace*\y+\prob*\a) -- (\xspace*\x + 0.5*\xspace, \zoomindest) -- (\xspace*\x + \xspace - \xsepintexpend, \zoomindest);
            }

            \foreach \num/\den [count=\j] in \triple{

                \tikzmath{
                    \start = 0;
                    \prob = 0;
                }
                \ifthenelse{\isodd{\x}} {
                    \tikzmath{\prob = \bprob;}
                } {
                    \tikzmath{\prob = \aprob;}
                }
                \ifthenelse{\equal{\j}{2}}{
                    \tikzmath{\start = \prob;}
                } {
                    \ifthenelse{\equal{\j}{3}} {
                        \tikzmath{\start = 1;}
                    }
                }
                        
                \tikzmath{\xoffset = \xnumspace;}
                \ifthenelse{\num > 99} {
                    \tikzmath{\xoffset = \xoffset + 0.06;}
                    \ifthenelse{\num > 999} {
                        \tikzmath{\xoffset = \xoffset + 0.06;}
                    } {}
                } {}

                \ifthenelse{\( \equal{\j}{1} \AND \equal{\charchosen}{1} \) \OR \( \equal{\j}{3} \AND \equal{\charchosen}{0} \)} {
                    \ifthenelse{\equal{\den}{0}} {
                        \node () at (\xspace*\x - \xoffset, \yspace*\y + \start*\a) {$\num$};
                    } {
                        \ifthenelse{\equal{\den}{1}} {
                            \node () at (\xspace*\x - \xoffset, \yspace*\y + \start*\a) {$\frac{\num}{2}$};
                        } {
                        \node () at (\xspace*\x - \xoffset, \yspace*\y + \start*\a) {$\frac{\num}{2^{\den}}$};
                        }
                    }
                } { 
                    \ifthenelse{\equal{\den}{0}} {
                        \node () at (\xspace*\x - \xoffset, \yspace*\y + \start*\a) {$\mathbf{\num}$};
                    } {
                        \ifthenelse{\equal{\den}{1}} {
                            \node () at (\xspace*\x - \xoffset, \yspace*\y + \start*\a) {$\mathbf{\frac{\num}{2}}$};
                        } {
                        \node () at (\xspace*\x - \xoffset, \yspace*\y + \start*\a) {$\mathbf{\frac{\num}{2^{\den}}}$};
                        }
                    }
                }
                
            } 

            \ifthenelse{\isodd{\i}}{
                \node () at (\xspace*\x - \subsdistx,\yspace*\y - \subsdisty) {\scriptsize $u_{\nodeid}$ char $a$};
            }{
                \node () at (\xspace*\x - \subsdistx,\yspace*\y - \subsdisty) {\scriptsize $u_{\nodeid}$ char $b$};
            }

        } 

        \end{scope}

	\end{tikzpicture}

\caption{
The figure shows an execution of our arithmetic coding for tries. 
The input trie $\mathcal{T}$ to be compressed is in the top left corner and $\Sigma=\{a,b\}$. 
Since $\mathcal{T}$ has $4$ nodes, $2$ edges labeled by $a$, and $1$ edge labeled by $b$, the (unconditioned) probabilities are $p(a) = \frac{1}{2} =  p(\lnot a)$, $p(b)=\frac{1}{4}$ and $p(\lnot b)=\frac{3}{4}$. 
By applying Algorithm~\ref{arithmetic compression} for $k = 0$, the encoder computes the sequence of intervals shown on the right.
The final interval $[l,l+s)$, where $l = \frac{115}{2^7}$ and $s = \frac{27}{2^{12}}$ is shown in the bottom right corner.
The encoder emits the binary sequence $111001101$ corresponding to the mid point of $[l,l+s)$ in binary, truncated to its first $d=9$ most significant bits.
}\label{fig:arithmetic_coding}
\end{figure}

\subsection{Comparison with the label entropy}
A different notion of empirical entropy for ordered node-labeled trees is the \emph{label entropy} introduced by Ferragina et al.~\cite{xBWTconf}.
This entropy has been proposed as a measure for the compressibility of the characters labeling the nodes in a tree based on their context.
In the following we report a trivial adaptation of their formula to ordered \emph{edge}-labeled trees.

\begin{definition}[Label entropy]\cite[Section 6]{xBWTconf}\label{def: label entropy}
    Let $cover(w)$ be the string obtained by concatenating (in any order) the labels outgoing from the nodes having the same length-$k$ context $w\in\Sigma^k$.
    The (unnormalized) $k$-th order label entropy $\mathcal{H}^{label}_k$ for a tree $\mathcal{T}$ is defined as $\mathcal{H}^{label}_k(\mathcal{T}) = \sum_{w \in \Sigma^k} |cover(w)|\mathcal{H}_0(cover(w))$ where $\mathcal{H}_0(cover(w))$ is the $0$-th order empirical entropy of the string $cover(w)$.
\end{definition}

It easy to observe that, due to the log-sum inequality also $\mathcal H_k^{label}(\mathcal T)$ is non-increasing when $k$ increases.
Furthermore, we recall that for $k=0$ the label entropy $H_0^{label}(\mathcal{T})$ coincides with the (unnormalised) $0$-th order empirical entropy $H_0(S)$ of a string $S$ obtained by concatenating in any order the labels appearing in $\mathcal{T}$.
We note that $\mathcal{H}^{label}_k(\mathcal{T})$ and $\mathcal{H}_k(\mathcal{T})$ have different normalization factors which are the number $n-1$ of labels and the number $n$ of nodes, respectively.
We observe the following main differences between the above two entropies.
While $\mathcal{H}_{k}^{label}$ was originally defined for labeled ordered trees and considers only the labels appearing in the tree, our entropy $\mathcal{H}_k$ of Definition~\ref{def: k-entropy} is specifically designed for the subclass of tries and differently from $\mathcal{H}^{label}_k$ it encodes also the topology of the trie, indeed, as shown in Theorem~\ref{theorem: reachability} every trie $\mathcal{T}$ can be represented within $n\mathcal{H}_k(\mathcal{T})+2$ bits of space if we assume that the symbols statistics are already stored. 
In the following, we compare the two measures $\mathcal H_k(\mathcal T)$ and $\mathcal H_k^{label}(\mathcal T)$ when both are applied to a trie $\mathcal T$.
Indeed, although $\mathcal H_k^{label}(\mathcal T)$ was originally introduced for trees, different works in the literature (see for instance~\cite{HonTrieLabelEntropy, Kosolobov2019}) used this measure to analyse the space occupation of compressed representations for tries.
We start this comparison by recalling a result already proved in~\cite[Section 4]{Kosolobov2019}; for the sake of completeness and uniformity, below we report their proof rewritten using our notation.

\begin{lemma}\cite[Section 4]{Kosolobov2019}\label{lemma: Kosolobov2019}
$\sum_{w \in  \Sigma^k}\sum_{c \in \Sigma} n_{w,c}\log ( \frac{n_w}{n_{w,c}})\leq (n-1)\mathcal{H}^{label}_k(\mathcal{T})+ \log e$ holds for every trie $\mathcal{T}$
\end{lemma}
\begin{proof} 
Let $m_w = \sum_{c \in \Sigma} n_{w,c}$ be the total number of edges outgoing from the nodes with context $w$. 
We observe that $m_w=\lvert cover(w)\rvert $ and $n_{w,c}$ is the number of times the symbol $c$ occurs in $cover(w)$. 
Therefore the (unnormalized) label entropy $(n-1)\mathcal{H}_{k}^{label}(\mathcal{T})$ can be rewritten as $\sum_{w \in  \Sigma^k}\sum_{c \in \Sigma} n_{w,c}\log \frac{m_w}{n_{w,c}}$. 
Now we compute how much it can increase if we replace $m_w$ with $n_w$ in its definition. 
We start from $\sum_{w \in  \Sigma^k}\sum_{c \in \Sigma} n_{w,c}\log \frac{n_w}{n_{w,c}}$, by rewriting $\frac{n_w}{n_{w,c}}$ as $\frac{m_{w}}{n_{w,c}}+\frac{n_w-m_w}{n_{w,c}}$ we obtain $\sum_{w \in  \Sigma^k}\sum_{c \in \Sigma} n_{w,c}\log (\frac{m_w}{n_{w,c}}+\frac{n_w-m_w}{n_{w,c}})$. 
Now since $\log x$ is a concave and derivable function it holds that $\log(x+y)\leq \log x + \frac{y}{x}\log e$ for any real numbers $x,y$ such that $x>0$ and $x+y>0$, since $\frac{1}{x}\log e$ is the derivative of $\log x$.
Therefore the previous formula is upper-bounded by $\sum_{w \in  \Sigma^k}\sum_{c \in \Sigma} n_{w,c}\log (\frac{m_w}{n_{w,c}})+(\log e \sum_{w \in  \Sigma^k}\sum_{c \in \Sigma}n_{w,c}\frac{n_w-m_w}{m_w})$. The claimed result follows because the left summation is $(n-1)\mathcal{H}_{k}^{label}(\mathcal{T})$ while for the second one it holds that $\sum_{w \in  \Sigma^k}\sum_{c \in \Sigma}n_{w,c}\frac{n_w-m_w}{m_w} = 1$ because $\sum_{c \in \Sigma}n_{w,c}=m_w$ while $\sum_{w \in  \Sigma^k} n_w=n$ and $\sum_{w \in  \Sigma^k} m_w=n-1$.
\end{proof}

The next result easily follows from the previous Lemma, and prove that our trie entropy can be larger than the label entropy (both unnormalised) by at most an $1.443n$ additive factor.
We recall that $(n-1)\mathcal{H}^{label}_k(\mathcal{T})$ considers only the information about the labels and not the topology of the trie.

\begin{corollary}\label{cor: comparison label entropy}
For every trie $\mathcal{T}$ it holds that $n\mathcal{H}_k(\mathcal{T})\leq (n-1)\mathcal{H}^{label}_k(\mathcal{T})+1.443n$
\end{corollary}
\begin{proof}
It is known that given a bitvector $B$ of size $n$ with $m$ ones it holds that $n\mathcal{H}_0(B) \leq m \log \frac{n}{m}+m\log e$~\cite[Sect. 2.3.1]{compactDataStructures}.
We note that each additive term $n_{w,c}\log \left(\frac{n_w}{n_{w,c}} \right) + (n_w - n_{w,c})\log \left( \frac{n_w}{n_w-n_{w,c}}\right )$ appearing in the definition of $n\mathcal{H}_k(\mathcal{T})$ corresponds to the (unnormalised) $0$-th order entropy of a bitvector $B$ of length $n_w$ with $n_{w,c}$ ones. 
Therefore it holds that $n\mathcal{H}_k(\mathcal{T}) \leq \sum_{w \in  \Sigma^k}\sum_{c \in \Sigma} n_{w,c}\log \frac{n_w}{n_{w,c}}+(n-1)\log e$ where the rightmost term derives from $\sum_{w \in  \Sigma^k}\sum_{c \in \Sigma} n_{w,c} = n-1$. 
Now by upper-bounding the left summation using Lemma~\ref{lemma: Kosolobov2019} and $\log e$ with $1.443$ we obtain the claimed result.
\end{proof}

Consider the worst-case entropy of ordered unlabeled trees, which is given by the formula $2n - \Theta(\log n)$~\cite[Section 2.1]{compactDataStructures}. 
Due to this latter result, we observe that the quantity $(n-1)\mathcal H_k^{label}(\mathcal T) + 2n - \Theta(\log n)$, which corresponds to compressing the labels of the trie and storing its tree topology separately within its worst-case space, is asymptotically lower-bounded by our $k$-th order empirical entropy $n\mathcal H_k(\mathcal T)$, for every possible trie $\mathcal T$.
In addition to that, in the next proposition we exhibit an infinite family of tries for which $(n-1)\mathcal H_k^{label}(\mathcal T) = \Omega(n)$ holds for every arbitrary large integer $k \geq 0$, while $n\mathcal H_k(\mathcal T) =0$, for every $k$ strictly larger than $0$.
See Figure~\ref{fig: trie confronto entropie} for an example.

\begin{proposition}\label{prop: infinite family of tries}
    For every arbitrarily large integers $n$ and $y$ satisfying $n = \sum_{i=0}^{h} y^i$, for some integer $h\geq 0$, there exists a trie $\mathcal T$ of $n$ nodes and height $h$ satisfying the following properties.
    \begin{enumerate}
        \item $n\mathcal H_k(\mathcal T) = 0$ for every $k$ with $k \geq 1$.
        \item $(n-1)\mathcal H_k^{label}(\mathcal T) \geq (n-1)\log y$ for every non-negative integer $k$.
    \end{enumerate}
\end{proposition}

\begin{proof}
Consider the (unlabeled) complete and balanced $y$-ary tree $\mathcal T'$ of $n$ nodes, height $h=\log (n+1)-1$ and an alphabet $\Sigma$ of size $\sigma=hy$, not including the special character $\#$.
We build the trie $\mathcal T$ by labeling the edges of $\mathcal{T'}$ as follows.
We partition $\Sigma$ into $h$ disjoint subsets $\Sigma_d$ with $0\leq d < h$ each of size $y$.
Then for every internal node $u$ at depth $d$ we set $out(u)=\Sigma_d$.
In the following we show that $n\mathcal H_1(\mathcal T) = 0$, this implies the claimed result for the higher orders since $n\mathcal H_{k+1}(\mathcal T) \leq n\mathcal H_k(\mathcal T)$ and $n\mathcal H_k(\mathcal T)\geq 0$ hold for every $k\geq 0$.
The possible contexts of length $k=1$ are $w=\#$ and $w\in \Sigma=\bigcup \Sigma_d$ and in particular by the way $\mathcal{T}$ is built the following observations hold: 1) the only node reached by $w=\#$ is the root of the trie; and 2) if instead $w\in \Sigma_d$ for some $d$ then all the nodes with context $w$ are at depth $d+1$. 
These observations imply that for every context $w$ of length $k=1$, all the nodes $u$ with $\lambda_1(u)=w$ share the same set of outgoing labels and thus for every $c\in \Sigma$ it holds that $n_{w,c}=n_w$ or $n_{w,c}=0$.
As a consequence the contribution of each 
$\frac{n_{w,c}}{n}\log\left(\frac{n_w}{n_{w,c}}\right) + \frac{n_w - n_{w,c}}{n}\log( \frac{n_w}{n_w-n_{w,c}})$ term to $\mathcal H_1(\mathcal T)$ is $0$ for every $w,c \in \Sigma$ and thus $n\mathcal H_1(\mathcal T)=0$.

Now we prove that $(n-1)\mathcal H_k^{label}(\mathcal T) \geq (n-1)\log y$ for every $k \geq 0$.
Since in a trie every node is reached by a distinct string from the root, we know that for every integer $k$, with $k \geq h$, and for every string $w \in \Sigma^k$, there exists at most a node $u$ in $\mathcal T$, such that $\lambda_k(u) = w$.
Therefore, due to the construction of $\mathcal T$, we know that $cover(w)$ is either the empty string $\epsilon$ or is a string of length $y$ formed by all the characters of $\Sigma_d$, where $d$ is the depth of the unique internal node $u$ such that $\lambda_k(u)  = w$.
Consequently, we can observe that $\mathcal{H}_0(cover(w)) = 0$ if $cover(w) = \epsilon$ and $\mathcal{H}_0(cover(w)) = \log y$ otherwise, since in this latter case we have that all the $y$ characters of $cover(w)$ are distinct.
Therefore, for every $k \geq h$ arbitrarily large, it follows that $(n-1)\mathcal H_k^{label}(\mathcal T) = \sum_{w \in \Sigma^k}\lvert cover(w) \rvert \mathcal H_0(cover(w)) = \sum_{u\in I} \lvert out(u)\rvert \log y = (n-1)\log y$, where $I$ is the set of internal nodes in $\mathcal{T}$.
Finally, since $\mathcal H_{k+1}^{label}(\mathcal T) \leq \mathcal H_{k}^{label}(\mathcal T)$ for every integer $k \geq 0$, the final result follows.
\end{proof}

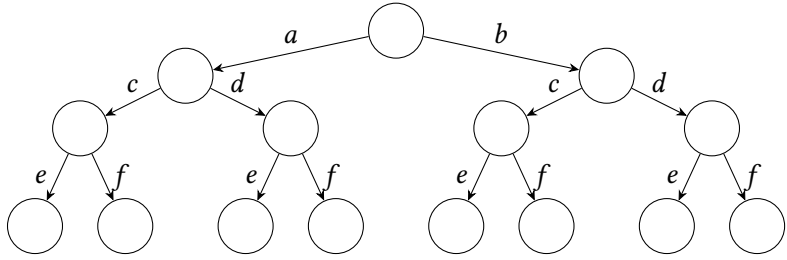
\begin{figure}
    \centering
\begin{tikzpicture}[
        >={Stealth[length=1.4mm, width=1.2mm]},
        dim/.style={minimum size=1.9em, font={\small}},
    ]

        \node[state, dim] (1) at (4.9,2.8) {};

        \node[state, dim] (2) at (2.1,2.2) {};
        \node[state, dim] (3) at (7.7,2.2) {};

        \node[state, dim] (4) at (0.7,1.5) {};
        \node[state, dim] (5) at (3.5,1.5) {};
        \node[state, dim] (6) at (6.3,1.5) {};
        \node[state, dim] (7) at (9.1,1.5) {};
        
		\node[state, dim] (8) at (0.1,0.2) {};
        \node[state, dim] (9) at (1.3,0.2) {};
        \node[state, dim] (10) at (2.9,0.2) {};
        \node[state, dim] (11) at (4.1,0.2) {};
        \node[state, dim] (12) at (5.7,0.2) {};
        \node[state, dim] (13) at (6.9,0.2) {};
        \node[state, dim] (14) at (8.5,0.2) {};
        \node[state, dim] (15) at (9.7,0.2) {};
		
		\draw[->] (1) to node [above] {$a$} (2);
		\draw[->] (1) to node [above] {$b$} (3);
        
		\draw[->] (2) to node [above] {$c$} (4);
        \draw[->] (2) to node [above] {$d$} (5);
		\draw[->] (3) to node [above] {$c$} (6);
		\draw[->] (3) to node [above] {$d$} (7);

        \draw[->] (4) to node [left] {$e$} (8);
        \draw[->] (4) to node [right] {$f$} (9);
		\draw[->] (5) to node [left] {$e$} (10);
		\draw[->] (5) to node [right] {$f$} (11);
        \draw[->] (6) to node [left] {$e$} (12);
        \draw[->] (6) to node [right] {$f$} (13);
		\draw[->] (7) to node [left] {$e$} (14);
		\draw[->] (7) to node [right] {$f$} (15);
\end{tikzpicture}
    \caption{The figure shows a trie $\mathcal T$ belonging to the family of tries described in Proposition~\ref{prop: infinite family of tries}. In this example, we have that $n = 15$, $y = 2$, and $h = 3$, moreover the alphabet partition is $\Sigma_0 = \{a,b\}$, $\Sigma_1 = \{c,d\}$, and $\Sigma_2 = \{e,f\}$.
    For this trie $(n-1)\mathcal H_k^{label}(\mathcal T) \geq (n - 1)\log y = 14$ holds for every $k \geq 0$, while $n\mathcal H_k(\mathcal T) = 0$, for every $k \geq 1$.}
    \label{fig: trie confronto entropie}
\end{figure}

\section{The XBWT of a trie}\label{section: XBWT}

The XBWT~\cite{xBWTjournal} introduced by Ferragina et al., is an extension of the original Burrows-Wheeler transform for strings~\cite{BWT} to ordered node-labeled trees.
Analogously to its string counterpart, the XBWT can be employed to simultaneously compress and index the input tree~\cite{xBWTjournal}.
In this section, we show the existence of a representation for the XBWT of a trie $\mathcal T$ using a number of bits upper-bounded by $n\mathcal H_k(\mathcal T) + o(n)$ for every sufficiently small integer $k$ \emph{simultaneously}.
In the following, we consider the definition of the XBWT, denoted hereafter by $\BWT(\mathcal T)$, proposed by Prezza for the more specific case of tries~\cite[Definition 3.1]{rindexTrie}.
Let $r$ be the root of the trie $\mathcal{T}$, and let $\lambda(r \rightsquigarrow u)$ be the string labeling the downward path from $r$ to $u$.
Formally, if $u$ has depth $d$ then $\lambda(r\rightsquigarrow u)=\lambda_d(u)$.
Moreover, let $u_1,u_2,\ldots, u_n$ be the sequence of $\mathcal T$ nodes sorted \emph{co-lexicographically}, i.e., the nodes are sorted according to the co-lexicographic order of the string labeling their incoming path from the root.
Formally, for any two arbitrary nodes $u_i$ and $u_j$, we have that $i \leq j$ holds if and only if $\lambda(r\rightsquigarrow u_i)\preceq \lambda(r\rightsquigarrow u_j)$.
The XBWT of a trie $\BWT(\mathcal T)$~\cite{rindexTrie}, is informally defined as the sequence of $n$ sets, where the $i$-th set contains the labels outgoing from the $i$-th node $u_i$ according to the co-lexicographic order.

\begin{definition}[XBWT of a trie]\cite[Definition 3.1]{rindexTrie}\label{definition: Burrows-Wheeler}
Given a trie $\mathcal{T}$ with co-lexicographic sorted nodes $u_1,\ldots,u_n$, then $
\BWT(\mathcal{T})=out(u_1),\ldots,out(u_n)$
\end{definition}

Figure~\ref{fig: example XBWT} shows an example of the XBWT of a trie.
The XBWT is an invertible transformation and it can be computed in $O(n)$ time using the algorithm proposed by Ferragina et al.~\cite[Theorem 2]{xBWTjournal}.
It is easy to see that this transformation can be represented using $\sigma$ bitvectors $B_c$ of length $n$ each.
Indeed, to store the $\BWT(\mathcal T)$ we can simply set the $i$-th bit of $B_c$ to $1$ if the $i$-th node in co-lexicographic order has an outgoing edge labeled by the character $c$.
Formally, for every character $c \in \Sigma$ and integer $i \in [n]$, we set $B_c[i] = 1$ if and only if $c \in out(u_i)$.
This representation of the XBWT of a trie based on the aforementioned $\sigma$ bitvectors $B_c$ was (implicitly) proposed by Belazzougui~\cite[Theorem 2]{BelazzouguiAhoCorasick} to compress the trie corresponding to the \texttt{next} transition function of the Aho-Corasick automaton~\cite{AC75} for a set of strings. 
In~\cite[Theorem 2]{BelazzouguiAhoCorasick}, Belazzougui showed that, when $\sigma\leq n^{\varepsilon}$ for some constant $0 < \varepsilon < 1$, by compressing each bitvector $B_c$ with the ID representation of Lemma~\ref{thm:RRR}, a trie $\mathcal{T}$ can be represented within $(n-1)H_0^{label}(\mathcal{T})+1.443n+o(n)$ bits of space.
Consider now the array $C$ storing for every character $c \in \Sigma$ the number of nodes reached by a symbol strictly smaller than $c$. Formally, $C[c] = (\sum_{c' < c}n_{c'}) + 1$, where the plus $1$ term stems from the root.
As already noted in~\cite{BelazzouguiAhoCorasick}, if in addition to these IDs we store this array $C$, then within the same compressed space we can support in $O(1)$ time the $child(i,c)$ operation defined as follows. 
The operation $child(i,c)$ returns $-1$ if the node $u_i$ has no outgoing edge labeled $c$, otherwise it returns the co-lexicographic rank $j$ of the child node $u_j$ having incoming label $c$.
Indeed partial rank operations are supported in $O(1)$ time by the indexable dictionaries and, due to the properties of the XBWT~\cite{xBWTjournal}, it holds that $child(i,c) = -1$ iff $\prank(i,B_c)= -1$ and $child(i,c)=C[c]+\prank(i,B_c)$ otherwise~\cite{BelazzouguiAhoCorasick}. 
Later, Hon et al.~\cite{HonTrieLabelEntropy} made clear the connection between Belazzougui's representation and the XBWT and showed that the above representation of the $\BWT(\mathcal T)$ can be further compressed to $(n-1)H_k^{label}(\mathcal{T})+O(n)$ bits of space for any sufficiently small fixed $k$.
They obtained this result by partitioning every bitvector $B_c$ into sub-intervals obtained by grouping together the (contiguous) positions corresponding to the nodes reached by the same context of length $k$ and by individually compressing the corresponding bitvector portion.
Also this representation supports $child(i,c)$ in $O(1)$ time~\cite{HonTrieLabelEntropy}.
The next existing compressed representation of the $\BWT(\mathcal T)$ that we present is that of Kosolobov and Sivukhin~\cite[Section 4]{Kosolobov2019}.
In their solution~\cite[Lemma 6]{Kosolobov2019}, the authors show that by compressing every bitvector individually using the fixed block compression boosting technique~\cite{FixBlockCompressionBoostConf,FixBlockCompressionBoostJournal} one can reach $(n-1)H_k^{label}(\mathcal{T})+1.443n+o(n)$ bits of space for every sufficiently small $k$ \emph{simultaneously}. Later in this section we prove that the space occupation of their solution can be analysed in terms of our trie entropy and in particular it can be represented in at most $nH_k(\mathcal T)+o(n)$ bits of space, for every integer $k$ sufficiently small simultaneously.
Before diving into the details of the representation, we recall a result from~\cite{FixBlockCompressionBoostJournal}, which is of central importance to the fixed block compression boosting technique for strings~\cite{FixBlockCompressionBoostConf,FixBlockCompressionBoostJournal}, to the space analysis in~\cite[Lemma 6]{Kosolobov2019}, and to the proofs of Theorems~\ref{theorem: trieHk id} and \ref{theorem: trieHk fid}.
We recall that the measure $\mathcal{H}_0$ considered in the next result is the $0$-th order empirical entropy for strings.

\begin{lemma}\cite[Lemma 4]{FixBlockCompressionBoostJournal}\label{lemma: partitionedH0}
    Let $X_1\cdots X_\ell$ be an arbitrary partition of a string $X$ into $\ell$ blocks and let $X_1^b\cdots X^b_m$ be a partition of $X$ into $m$ blocks of size at most $b$, then $\sum_{i=1}^m\lvert X_i^b \rvert \mathcal{H}_0(X_i^b) \leq \sum_{i=1}^\ell\lvert X_i \rvert \mathcal{H}_0(X_i)+(\ell-1)b$
\end{lemma}

The compressed solution given by Kosolobov and Sivukhin~\cite[Section 4]{Kosolobov2019} is the following.
We partition every $B_c$ into $t$ blocks $B_c^{i}$ each of fixed size $b = \lceil \sigma \log^2 n\rceil$.  
Then we apply the ID representation of Lemma~\ref{thm:RRR} on every bitvector $B_c^i$ containing at least an entry equal to $1$, and we store the pointers to these IDs in a $\sigma \times \lceil n/b\rceil$ matrix.
For answering partial rank queries on every $B_c$ we also store the array $R[1..\sigma][1..\lceil n/b\rceil]$ which memorizes the precomputed rank values preceding every block $B_c^{i}$ for every $c$, that is $R[c][i]=rank(b(i-1),B_c)$.
Clearly, by using the arrays $R$ and $C$ it is possible to reconstruct whether a bitvector $B_c^i$ is empty (i.e., it contains only zeros) or not.
Moreover, through $R$ we can compute in $O(1)$ time $\prank(j, B_c)$ when $B_c[j]=1$ as $R[c][i] + \prank(j - (i-1)b, B_c^{i})$, where $i=\lceil j / b \rceil$ is the index of the interval where the position $j$ falls. Since this representation supports partial rank queries on every $B_c$ in $O(1)$ time, it can also answer $child(i,c)$ queries in constant time~\cite{Kosolobov2019}.
Note that both $R$ and the IDs pointers take $O(\sigma n/b \log n)$ bits of space which is $o(n)$.
Before proceeding further, we prove a preliminary result showing how this data structure can be used to compute full rank queries on the $\sigma$ bitvectors $B_c$.

\begin{proposition}\label{prop: full rank queries on id}
    The representation of $\BWT(\mathcal T)$ of Kosolobov and Sivukhin~\cite[Lemma 6]{Kosolobov2019} supports full rank queries $\rank(i, B_c)$ on the $\sigma$ bitvectors $B_c$, for every $i \in [n]$, in $O(\log \sigma  + \log\log n)$ time.
\end{proposition}

\begin{proof}
    We know that $\rank(j,B_c) = R[c][i] + \rank(j - (i-1)b, B_c^i)$, where $i = \lceil j/b\rceil$.
    The value $R[c][i]$ is stored explicitly and can be retrieved in constant time.
    On the other hand, by Remark~\ref{remark:IDrank} we know that the integer $\rank(j - (i-1)b, B_c^i)$ can be computed in $O(\log x_c^i)$, where $x_c^i$ is the number of ones in $B_c^i$, through a binary search on $B_c^i$.
    Clearly, it holds that $x_c^i \leq b = \lceil \sigma \log^2 n \rceil$, consequently $O(\log x_c^i) = O(\log \sigma + \log \log n)$.
\end{proof}

Thus, by assuming a uniform distribution over the $\sigma \times n$ positions of the bitvectors, if $t>1$, then we observe that since $b/n\sigma \leq 1/(t-1)\sigma$, the expected time for a full rank query is upper-bounded by $\sum_{(i,c) \in [t] \times \Sigma} \frac{1}{(t-1)\sigma} O(\log (x_c^i))$, which is equal to $\frac{t}{t-1}\sum_{(i,c) \in [t] \times \Sigma} \frac{1}{t\sigma} O(\log (x_c^i))$. 
Due to the Jensen's inequality, this is in turn upper-bounded by $O(\log(\frac{1}{t \sigma}\sum_{(i,c) \in [t] \times \Sigma} x_c^i)) = O(\log(\frac{n-1}{t\sigma}))  = O(\log \log n)$.
In the same paper\cite[Lemma 6]{Kosolobov2019}, Kosolobov and Sivukhin described how this representation can support select queries on the $\sigma$ bitvectors in $O(1)$ time.
This can be done through a bitvector $S$ formed by concatenating the unary encodings of the number of ones in each of the $\sigma t$ bitvectors $B_c^i$.
Formally, let $S_c$ be the bitvector formed by the bits $1^{x_{c}^1}01^{x_c^2}0\ldots1^{x_c^t}0$, where $x_c^i$ denotes the number of ones in the bitvector $B_c^i$.
Then $S = S_{c_1}S_{c_2}\ldots S_{c_\sigma}$, where $c_1, c_2, \ldots, c_\sigma$ are the $\sigma$ characters of $\Sigma$ in the order $c_1 \preceq c_2 \preceq \ldots \preceq c_\sigma$.
Kosolobov and Sivukhin showed that by applying the ID representation of Raman et al.~\cite[Theorem 4.6]{RRR} we can compress $S$ in $o(n)$ bits and support select queries, by using the array $C$, in $O(1)$ time~\cite[Lemma 6]{Kosolobov2019}.
Moreover, in~\cite[Lemma 6]{Kosolobov2019} the authors show that if $\sigma\leq n^{\varepsilon}$ for some constant $\varepsilon<1$, this compressed representation of the $\BWT(\mathcal T)$ takes at most $(n-1)H_k^{label}(\mathcal{T})+1.443n+o(n)$ bits of space for every non-negative $k\leq max\{0,\alpha \log_\sigma n-2\}$ \emph{simultaneously}, where $\alpha\in (0,1)$ is an arbitrary real valued constant.
In the following theorem, we show that it is possible to refine the space complexity of this data structure in terms of our empirical entropy $n\mathcal H_k(\mathcal T)$.

\begin{theorem}\label{theorem: trieHk id}
    Let $\mathcal{T}$ be a trie and $\varepsilon$ be an arbitrary constant with $0 \leq \varepsilon <1$.
    If $\sigma \leq n^\varepsilon$, then there exists an encoding of $\BWT(\mathcal T)$ occupying at most $n\mathcal H_k(\mathcal T) + o(n)$ bits of space simultaneously for every integer $k \leq \max\{0, \alpha \log _\sigma n - 2 \}$, where $\alpha$ is a constant satisfying $0 \leq \alpha < 1$.
    This data structure supports the operation $child(i,c)$ in $O(1)$ time for every integer $i \in [n]$ and character $c \in \Sigma$.
\end{theorem}

\begin{proof}
    As we have already discussed in this section how to perform the query $child(i,c)$ in $O(1)$ time, to prove the correctness of this theorem it remains to show the claimed spatial upper-bound.
    Consider the partition $\mathcal{V}_k$ of $V$ which groups together the nodes of $\mathcal{T}$ reached by the same length-$k$ context $w\in \Sigma^k$, that is $\mathcal{V}_k = \{P_w\neq \emptyset \mid w\in \Sigma^k\}$ where $P_w = \{ v\in V \mid \lambda_k(v)=w \}$.
    By Definition~\ref{definition: Burrows-Wheeler}, every part in $\mathcal{V}_k$ corresponds to a range of nodes in the co-lexicographic ordering.
    Given $\mathcal{V}_k=V_1,\ldots ,V_\ell$, we call $B_{w,c}$ the portion of $B_c$ whose positions correspond to the interval $V_i$ of nodes reached by the context $w \in \Sigma^k$.
    We observe that every $B_{w,c}$ has length $n_{w}$ and contains exactly $n_{w,c}$ bits equal to $1$, where $n_{w}$ and $n_{w,c}$ are the integers of Definition~\ref{def: integers nw and nwc}.
    As a consequence $\mathcal{H}_k(\mathcal{T})= \frac{1}{n}\sum_{w \in \Sigma^k} n_w\sum_{c \in \Sigma}\frac{n_{w,c}}{n_w}\log (\frac{n_w}{n_{w,c}}) + \frac{n_w - n_{w,c}}{n_w}\log(\frac{n_w}{n_w-n_{w,c}})$ can be rewritten as $\frac{1}{n} \sum_{c\in \Sigma} \sum_{w\in \Sigma^k} \lvert B_{w,c}\rvert \mathcal{H}_0(B_{w,c})$. 
    Therefore, analogously to the case of strings, we can compress every $B_{w,c}$ individually using a $0$-th order binary compressor to achieve exactly $n\mathcal{H}_k(\mathcal{T})$ bits of space.
    Next, we prove that, by using Lemma~\ref{lemma: partitionedH0}, we can approximate this space through the fixed-size bitvectors $B_{c}^i$ described in this section.
    
    Therefore, for every $c \in \Sigma$, consider the $t$ bitvectors $B_c^i$ of fixed size $b = \lceil \sigma \log^2 n \rceil$.
    By applying the ID representation of Lemma~\ref{thm:RRR} on every bitvector $B_c^i$ containing at least an entry equal to $1$, we obtain a total space in bits which is at most $\sum_{c\in \Sigma}\sum_{i=1}^t \lceil \log \binom{ \lvert B_c^i \rvert}{x^i_c} \rceil+ o(x^i_c) + O(\log \log \lvert B_c^i \rvert )$, where $x^i_c$ is the number of $1$s in $B_c^i$ and $t=\lceil n/b \rceil$.
    We also observe that $\sum_{c\in \Sigma}\sum_{i=1}^t \lceil \log \binom{\lvert B_c^i \rvert }{x^i_c} \rceil < \sigma \lceil\frac{n}{b}\rceil + \sum_{c \in \Sigma}\sum_{i=1}^t \log \binom{\lvert B_c^i \rvert }{x^i_c} = \sum_{c \in \Sigma}\sum_{i=1}^t \log \binom{\lvert B_c^i \rvert}{x^i_c} + o(n)$.
    Moreover, due to~\cite[Equation 11.40]{elementsOfInformationTheory}, the term $\sum_{c \in \Sigma}\sum_{i=1}^t \log \binom{\lvert B_c^i \rvert}{x^i_c}$ is in turn upper-bounded by $\sum_{c\in \Sigma}\sum_{i=1}^t \lvert B_{c}^i\rvert \mathcal{H}_0(B^i_c)$, thus by now we have proved that the total space is at most $(\sum_{c\in \Sigma}\sum_{i=1}^t \lvert B_{c}^i\rvert \mathcal{H}_0(B^i_c)+ o(x^i_c) + O(\log \log \lvert B_{c}^i\rvert)) + o(n)$ bits.
    At this point, we note that $\sum_{c\in \Sigma}\sum_{i=1}^t o(x^i_c)$ is $o(n)$ since summing over all the $1$s we are counting the number of edges in $\mathcal{T}$.
    Furthermore, the last term $\sum_{c\in \Sigma}\sum_{i=1}^t O(\log \log \lvert B_{c}^i\rvert)$ is at most $O(\sigma \frac{n}{b} \log \log n)=O(\frac{n\log \log n}{\log^2n})=o(n)$.
    Consequently, so far we proved that the space occupation is at most $\sum_{c\in \Sigma}\sum_{i=1}^t \lvert B_{c}^i\rvert \mathcal{H}_0(B^i_c) + o(n)$.
    Next, by applying Lemma~\ref{lemma: partitionedH0} to the partitioning corresponding to $\mathcal{V}_k$, we can upper-bound the left summation with $n\mathcal{H}_k(\mathcal{T}) + \sigma(\ell-1)b$.
    Now consider the last term $\sigma (\ell-1)b$, if $k=0$ then the optimal partitioning $\mathcal{V}_k$ consists of a single set containing all the nodes and therefore $\ell=1$ and $\sigma (\ell-1)b=0$, otherwise when $k\leq \alpha\log_\sigma n -2$ it holds that $\sigma (\ell-1)b = O(\sigma^{k+2}\log^2 n) =o(n)$.
    Let us now consider the table $R[1..\sigma][1..\lceil n/b\rceil]$, which memorizes the precomputed rank values preceding every block $B_c^{i}$ for every $B_c$.
    As already noted, both the table $R[1..\sigma][1..\lceil n/b\rceil]$, storing the precomputed rank values, and the $O(\sigma \frac{n}{b})$ pointers to the ID representations of the bitvectors $B_{c}^i$ use $o(n)$ additional bits.
    We finally observe that the partitions are chosen independently of $k$, therefore the theorem holds for every $k \leq \max\{0,\alpha\log_\sigma n - 2\}$ simultaneously.
\end{proof}

Note that, due to Corollary~\ref{cor: comparison label entropy}, it follows that the space usage $n \mathcal{H}_k(\mathcal T) + o(n)$ we proved, asymptotically lower-bounds the space $(n-1)H_k^{label}(\mathcal{T})+1.443(n-1)+o(n)$ shown by Kosolobov and Sivukhin~\cite[Lemma 6]{Kosolobov2019}.
In addition to that, Proposition~\ref{prop: infinite family of tries} shows that our formula can be \emph{asymptotically} smaller in some cases.
Next, we prove that the data structure of Theorem~\ref{theorem: trieHk id} can be computed in $O(n)$ expected time.

\begin{proposition}
    Let $\mathcal{T}$ be an arbitrary trie.
    Then the compressed representation of $\BWT(\mathcal T)$ analysed in Theorem~\ref{theorem: trieHk id} can be computed in $O(n)$ expected time.
\end{proposition}

\begin{proof}
    As a first step, we sort the nodes of $\mathcal T$ co-lexicographically in $O(n)$ time by using the algorithm of Ferragina et al.~\cite[Theorem 2]{xBWTjournal}.
    Thus, at the end of this step we obtain the co-lexicographically sorted sequence $u_1, u_2, \ldots, u_n$ of nodes.
    Next, we represent the $\sigma$ bitvectors $B_c$ encoding the XBWT of a trie as follows.
    For each character $c \in \Sigma$ we create a sorted sequence of integers $L_c$, initially empty.
    Then we scan the $n$ nodes $u_1, u_2, \ldots, u_n$ from left to right, and for every node $u_i$, we append the integer $i$ to the list $L_c$ if and only if $c \in out(u_i)$.
    This process trivially requires $O(n)$ time, since there are $n$ nodes and $n -1$ edges in $\mathcal T$.
    It is easy to observe that these sorted lists $L_c$ are a univocal representation of the $\sigma$ bitvectors $B_c$, and consequently also of $\BWT(\mathcal T)$.
    At this point, to apply the fixed block compression technique, we scan every $L_c$ and we split it into $t = \lceil n / b \rceil$ sorted sublists $L_c^i$ such that $L_c^i$ contains all and only the integers $x$, with $x \in [b]$, for which $x + b(i-1) \in L_c$ holds, where $b = \lceil \sigma \log^2 n \rceil$. 
    During the same scan we also fill the matrix $R$ and on every non-empty sublist $L_c^i$ we build the ID representation proposed by Raman et al.~\cite[Theorem 4.6]{RRR} and we store its corresponding pointer.
    Every ID representation can be built in $O(x_c^i)$ expected time~\cite[Section 8]{RRR}, where $x_c^i = \lvert L_c^i \rvert$, thus, it takes $O(n)$ total expected time to build all these ID representations, since $\sum_{c \in \Sigma}\sum_{i\in [t]} x_c^i = n - 1$.
    Due to an analogous reasoning, also the ID representation of the $S$ bitvector can be computed in $O(n)$ expected time, since $S$ contains exactly $n-1$ bits equal to $1$. 
    Finally, since $R$ has exactly $\sigma \frac{n}{b} = o(n)$ entries and also the pointers are $o(n)$, it follows that the total expected time complexity is $O(n)$.
\end{proof}

\paragraph{Prefix-matching.} In this paragraph we recall a result already shown in~\cite{BelazzouguiAhoCorasick, HonTrieLabelEntropy}, namely that the XBWT of a trie can be used to perform prefix-matching queries in optimal constant time per symbol.
Consider a dictionary $S$ formed by a set of strings, a well-known data structure to perform prefix matching queries on $S$ is the trie $\mathcal{T} = (V,E)$ corresponding to $S$.
Every distinct prefix of $S$, including the empty one $\epsilon$ and all the strings of $S$, uniquely corresponds to a node of $\mathcal T$. 
Given a node $u \in V$ and its corresponding prefix $p$, there exists an edge $(u,v,c) \in E$, for some $v \in V$ and $c \in \Sigma$, if and only if, (i)~the string $pc$ is a prefix of some string in $S$ and (ii)~$v$ is the node representing this prefix $pc$.
Let us now consider the nodes $u_1, u_2, \ldots, u_n$ of $\mathcal{T}$ sorted in co-lexicographic order (see Definition~\ref{definition: Burrows-Wheeler}).
By the properties of the XBWT, we know that if $child(i,c) = -1$, with $u_i = u$, then $\mathcal T$ does not have a node associated to the prefix $pc$, namely $pc$ is not a prefix of any string in $S$.
On the other hand, if $child(i,c) = j$, with $j \neq -1$, then we know that the $j$-th node in co-lexicographic order is the one corresponding to the prefix $pc$.
As a consequence, it follows that the next corollary is directly implied by Theorem~\ref{theorem: trieHk id}.

\begin{corollary}\label{cor: prefix matching}
    Let $\mathcal T$ be a trie representing a dictionary $S$ and $u_i$ be the node of $\mathcal T$ corresponding to a prefix $p$ of some string in $S$.
    If $\sigma \leq n^\varepsilon$ for some constant $\varepsilon$, with $0 \leq \varepsilon < 1$, then there exists a data structure taking at most $n\mathcal H_k(\mathcal T) + o(n)$ bits of space for every $k \leq \max\{0,\alpha\log_\sigma n -2 \}$ simultaneously, where $\alpha$ is a constant satisfying $0 \leq \alpha < 1$, such that given the co-lexicographic rank $i$, it supports for every $c \in \Sigma$ the following operations in $O(1)$ time.
    \begin{itemize}
        \item Determine if $pc$ is a prefix of some string in $S$.
        \item Return the co-lexicographic rank of the node associated to the prefix $pc$, if $pc$ is a prefix of some string in $S$.
    \end{itemize}
\end{corollary}

Therefore, it follows that in $n\mathcal H_k(\mathcal T) + o(n)$ bits of space, we can check whether there exists a string in $S$ prefixed by a given input pattern in $O(m)$ optimal time, where $m$ is the length of the pattern.

\begin{figure}
\renewcommand{\arraystretch}{0.9}
\resizebox{0.95\textwidth}{!}{
    \centering
\begin{tikzpicture}[
        >={Stealth[length=2mm, width=1.4mm]},
        arr/.style={line width = 0.21mm},
        vertex/.style={minimum size=2.2em, font={\small}, draw = black, line width = 0.2mm, circle},
        largenodes/.style={font={\large}}
    ]

        \node[vertex] (1) at (-0.2,0) {1};

        \node[vertex] (2) at (-1.5,-1) {2};
        \node[vertex] (13) at (1.1,-1) {13};

        \node[vertex] (3) at (-4,-1.8) {3};
        \node[vertex] (21) at (-2,-2.4) {21};
        \node[vertex] (4) at (1.8,-2.4) {4};
        \node[vertex] (25) at (3.6,-1.8) {25};

        \node[vertex] (22) at (-5.25,-3) {22};
        \node[vertex] (9) at (-3.5,-3.3) {9};
        \node[vertex] (17) at (-1.4,-3.8) {17};
        \node[vertex] (23) at (0,-3.3) {23};
        \node[vertex] (20) at (4.2,-3.6) {20};
        \node[vertex] (12) at (2.5,-3.7) {12};

        \node[vertex] (10) at (-5.9,-4.4) {10};
        \node[vertex] (18) at (-4.7,-4.65) {18};
        \node[vertex] (24) at (-3.2,-5.1) {24};
        \node[vertex] (5) at (-2.1,-5.8) {5};
        \node[vertex] (26) at (-1,-5.3) {26};
        \node[vertex] (11) at (0.3,-4.8) {11};
        \node[vertex] (8) at (4.7,-5.4) {8};
        \node[vertex] (19) at (1.6,-5) {19};
        \node[vertex] (16) at (3.2,-5.2) {16};

        \node[vertex] (14) at (-6.85,-5.4) {14};
        \node[vertex] (6) at (-5.7,-6) {6};
        \node[vertex] (27) at (-4.4,-6.4) {27};
        \node[vertex] (28) at (2.8,-6.4) {28};
        \node[vertex] (7) at (1.5,-6.5) {7};
        \node[vertex] (15) at (-0.2, -6.4) {15};
        
		\draw[->, arr] (1) edge [bend right = 30] node [above, largenodes] {$a$} (2);
		\draw[->, arr] (1) edge [bend left = 30] node [above, largenodes] {$b$} (13);

        \draw[->, arr] (2) edge [bend right = 20] node [above, largenodes] {$a$} (3);
        \draw[->, arr] (2) edge [bend right = 20] node [left, largenodes] {$c$} (21);
        \draw[->, arr] (13) edge [bend left = 20] node [right, largenodes] {$a$} (4);
        \draw[->, arr] (13) edge [bend left = 20] node [above, largenodes] {$c$} (25);

        \draw[->, arr] (3) edge [bend right = 25] node [above, largenodes] {$c$} (22);
        \draw[->, arr] (21) edge [bend right = 30] node [above, largenodes] {$a$} (9);
        \draw[->, arr] (21) edge [bend left = 20] node [right, largenodes] {$b$} (17);
        \draw[->, arr] (4) edge [bend right = 15] node [above, largenodes] {$c$} (23);
        \draw[->, arr] (25) edge [bend left = 15] node [right, largenodes] {$b$} (20);
        \draw[->, arr] (25) edge [bend left = 20] node [right, largenodes] {$a$} (12);

        \draw[->, arr] (22) edge [bend right = 20] node [left, largenodes] {$a$} (10);
        \draw[->, arr] (22) edge [bend left = 15] node [right, largenodes] {$b$} (18);
        \draw[->, arr] (9) edge [bend left = 10] node [right, largenodes] {$c$} (24);
        \draw[->, arr] (17) edge [bend right = 15] node [left, largenodes] {$a$} (5);
        \draw[->, arr] (17) edge [bend left = 10] node [right, largenodes] {$c$} (26);
        \draw[->, arr] (23) edge [bend left = 10] node [right, largenodes] {$a$} (11);
        \draw[->, arr] (20) edge [bend left = 15] node [right, largenodes] {$a$} (8);
        \draw[->, arr] (23) edge [bend left = 20] node [above, largenodes] {$b$} (19);
        \draw[->, arr] (12) edge [bend left = 20] node [right, largenodes] {$b$} (16);

        \draw[->, arr] (10) edge [bend right = 30] node [above, largenodes] {$b$} (14);
        \draw[->, arr] (18) edge [bend right = 20] node [left, largenodes] {$a$} (6);
        \draw[->, arr] (18) edge [bend left = 10] node [right, largenodes] {$c$} (27);
        \draw[->, arr] (19) edge [bend left = 10] node [right, largenodes] {$c$} (28);
        \draw[->, arr] (19) edge [bend left = 10] node [left, largenodes] {$a$} (7);
        \draw[->, arr] (11) edge [bend left = 15] node [right, largenodes] {$b$} (15);

\end{tikzpicture}
}
\vspace{2
em}

\setlength{\tabcolsep}{0.35em}
\resizebox{1\textwidth}{!}{
\begin{tabular}{ | C{2.5em} | C{0.3em} | C{0.3em} | C{0.3em} | C{0.3em} | C{0.3em} | C{0.3em} | C{0.3em} | C{0.3em} | C{0.3em} | C{0.3em} | C{0.3em} | C{0.3em} | C{0.3em} |C{0.3em} |C{0.3em} |C{0.3em} | C{0.3em} | C{0.3em} | C{0.3em} | C{0.3em} | C{0.3em} | C{0.3em} | C{0.3em} | C{0.3em} | C{0.3em} | C{0.3em} | C{0.3em} | C{0.3em} | }

    \multicolumn{1}{C{2.5em}}{} & \multicolumn{1}{C{0.3em}}{} & \multicolumn{1}{C{0.3em}}{} & \multicolumn{1}{C{0.3em}}{} & \multicolumn{1}{C{0.3em}}{} & \multicolumn{1}{C{0.3em}}{} & \multicolumn{1}{C{0.3em}}{a} & \multicolumn{1}{C{0.3em}}{b} & \multicolumn{1}{C{0.3em}}{} & \multicolumn{1}{C{0.3em}}{} & \multicolumn{1}{C{0.3em}}{} & \multicolumn{1}{C{0.3em}}{} & \multicolumn{1}{C{0.3em}}{} & \multicolumn{1}{C{0.3em}}{} & \multicolumn{1}{C{0.3em}}{a} & \multicolumn{1}{C{0.3em}}{b} & \multicolumn{1}{C{0.3em}}{} & \multicolumn{1}{C{0.3em}}{} & \multicolumn{1}{C{0.3em}}{} & \multicolumn{1}{C{0.3em}}{} & \multicolumn{1}{C{0.3em}}{} & \multicolumn{1}{C{0.3em}}{} & \multicolumn{1}{C{0.3em}}{} & \multicolumn{1}{C{0.3em}}{} & \multicolumn{1}{C{0.3em}}{} & \multicolumn{1}{C{0.3em}}{} & \multicolumn{1}{C{0.3em}}{} & \multicolumn{1}{C{0.3em}}{a} & \multicolumn{1}{C{0.3em}}{b} \\ [-3px]    
    
    \multicolumn{1}{C{2.5em}}{} & \multicolumn{1}{C{0.3em}}{} & \multicolumn{1}{C{0.3em}}{} & \multicolumn{1}{C{0.3em}}{} & \multicolumn{1}{C{0.3em}}{} & \multicolumn{1}{C{0.3em}}{a} & \multicolumn{1}{C{0.3em}}{a} & \multicolumn{1}{C{0.3em}}{a} & \multicolumn{1}{C{0.3em}}{b} & \multicolumn{1}{C{0.3em}}{} & \multicolumn{1}{C{0.3em}}{a} & \multicolumn{1}{C{0.3em}}{b} & \multicolumn{1}{C{0.3em}}{} & \multicolumn{1}{C{0.3em}}{} & \multicolumn{1}{C{0.3em}}{a} & \multicolumn{1}{C{0.3em}}{a} & \multicolumn{1}{C{0.3em}}{b} & \multicolumn{1}{C{0.3em}}{} & \multicolumn{1}{C{0.3em}}{a} & \multicolumn{1}{C{0.3em}}{b} & \multicolumn{1}{C{0.3em}}{} & \multicolumn{1}{C{0.3em}}{} & \multicolumn{1}{C{0.3em}}{} & \multicolumn{1}{C{0.3em}}{} & \multicolumn{1}{C{0.3em}}{a} & \multicolumn{1}{C{0.3em}}{} & \multicolumn{1}{C{0.3em}}{a} & \multicolumn{1}{C{0.3em}}{a} & \multicolumn{1}{C{0.3em}}{a} \\ [-3px]

    \multicolumn{1}{C{2.5em}}{} & \multicolumn{1}{C{0.3em}}{} & \multicolumn{1}{C{0.3em}}{} & \multicolumn{1}{C{0.3em}}{} & \multicolumn{1}{C{0.3em}}{} & \multicolumn{1}{C{0.3em}}{c} & \multicolumn{1}{C{0.3em}}{c} & \multicolumn{1}{C{0.3em}}{c} & \multicolumn{1}{C{0.3em}}{c} & \multicolumn{1}{C{0.3em}}{a} & \multicolumn{1}{C{0.3em}}{a} & \multicolumn{1}{C{0.3em}}{a} & \multicolumn{1}{C{0.3em}}{b} & \multicolumn{1}{C{0.3em}}{} & \multicolumn{1}{C{0.3em}}{c} & \multicolumn{1}{C{0.3em}}{c} & \multicolumn{1}{C{0.3em}}{c} & \multicolumn{1}{C{0.3em}}{a} & \multicolumn{1}{C{0.3em}}{a} & \multicolumn{1}{C{0.3em}}{a} & \multicolumn{1}{C{0.3em}}{b} & \multicolumn{1}{C{0.3em}}{} & \multicolumn{1}{C{0.3em}}{a} & \multicolumn{1}{C{0.3em}}{b} & \multicolumn{1}{C{0.3em}}{c} & \multicolumn{1}{C{0.3em}}{} & \multicolumn{1}{C{0.3em}}{c} & \multicolumn{1}{C{0.3em}}{c} & \multicolumn{1}{C{0.3em}}{c} \\ [-3px]


    \multicolumn{1}{C{2.5em}}{}  & \multicolumn{1}{C{0.3em}}{} & \multicolumn{1}{C{0.3em}}{} & \multicolumn{1}{C{0.3em}}{a} & \multicolumn{1}{C{0.3em}}{b} & \multicolumn{1}{C{0.3em}}{b} & \multicolumn{1}{C{0.3em}}{b} & \multicolumn{1}{C{0.3em}}{b} & \multicolumn{1}{C{0.3em}}{b} & \multicolumn{1}{C{0.3em}}{c} & \multicolumn{1}{C{0.3em}}{c} & \multicolumn{1}{C{0.3em}}{c} & \multicolumn{1}{C{0.3em}}{c} & \multicolumn{1}{C{0.3em}}{} & \multicolumn{1}{C{0.3em}}{a} & \multicolumn{1}{C{0.3em}}{a} & \multicolumn{1}{C{0.3em}}{a} & \multicolumn{1}{C{0.3em}}{c} & \multicolumn{1}{C{0.3em}}{c} & \multicolumn{1}{C{0.3em}}{c} & \multicolumn{1}{C{0.3em}}{c} & \multicolumn{1}{C{0.3em}}{a} & \multicolumn{1}{C{0.3em}}{a} & \multicolumn{1}{C{0.3em}}{a} & \multicolumn{1}{C{0.3em}}{a} & \multicolumn{1}{C{0.3em}}{b} & \multicolumn{1}{C{0.3em}}{b} & \multicolumn{1}{C{0.3em}}{b} & \multicolumn{1}{C{0.3em}}{b} \\ [-3px]
    
    \multicolumn{1}{C{0.3em}}{} & \multicolumn{1}{C{0.3em}}{$\epsilon$} & \multicolumn{1}{C{0.3em}}{a} & \multicolumn{1}{C{0.3em}}{a} & \multicolumn{1}{C{0.3em}}{a} & \multicolumn{1}{C{0.3em}}{a} & \multicolumn{1}{C{0.3em}}{a} & \multicolumn{1}{C{0.3em}}{a} & \multicolumn{1}{C{0.3em}}{a} & \multicolumn{1}{C{0.3em}}{a} & \multicolumn{1}{C{0.3em}}{a} & \multicolumn{1}{C{0.3em}}{a} & \multicolumn{1}{C{0.3em}}{a} & \multicolumn{1}{C{0.3em}}{b} & \multicolumn{1}{C{0.3em}}{b} & \multicolumn{1}{C{0.3em}}{b} & \multicolumn{1}{C{0.3em}}{b} & \multicolumn{1}{C{0.3em}}{b} & \multicolumn{1}{C{0.3em}}{b} & \multicolumn{1}{C{0.3em}}{b} & \multicolumn{1}{C{0.3em}}{b} & \multicolumn{1}{C{0.3em}}{c} & \multicolumn{1}{C{0.3em}}{c} & \multicolumn{1}{C{0.3em}}{c} & \multicolumn{1}{C{0.3em}}{c} & \multicolumn{1}{C{0.3em}}{c} & \multicolumn{1}{C{0.3em}}{c} & \multicolumn{1}{C{0.3em}}{c} & \multicolumn{1}{C{0.3em}}{c} \\ 
               
    \hline
    co-lex & \tiny{1} & \tiny{2} & \tiny{3} & \tiny{4} & \tiny{5} & \tiny{6} & \tiny{7} & \tiny{8} & \tiny{9} & \tiny{10} & \tiny{11} & \tiny{12} & \tiny{13} & \tiny{14} & \tiny{15} & \tiny{16} & \tiny{17} & \tiny{18} & \tiny{19} & \tiny{20} & \tiny{21} & \tiny{22} & \tiny{23} & \tiny{24} & \tiny{25} & \tiny{26} & \tiny{27} & \tiny{28} \\
    \hline
                       & a & \textcolor{red}{a} &   &   &   &   &   &   &   &   &   &   & \textcolor{red}{a} &   &   &   & a & a & a & a & a & a & \textcolor{red}{a} &   & \textcolor{red}{a} &   &   &   \\
    $\BWT$ & \textcolor{red}{b} &   &   &   &   &   &   &   &   & b & b & \textcolor{red}{b} &   &   &   &   &   &   &   &   & b & b & \textcolor{red}{b} &   & \textcolor{red}{b} &   &   &   \\
                       &   & c & c & \textcolor{red}{c} &   &   &   &   & \textcolor{red}{c} &   &   &   & \textcolor{red}{c} &   &   &   & c & c & \textcolor{red}{c} &   &   &   &   &   &   &   &   &   \\
    \hline
\end{tabular}
} 

\caption{In the upper part, the figure shows a trie $\mathcal T$ of $28$ nodes over an alphabet of size $3$.
The integer associated to every node represents its position in the co-lexicographic order.
The table at the bottom shows the incoming string $\lambda(r \rightsquigarrow u)$ connecting the node $u$ to the root $r$, as well as the XBWT transform of the entire trie $\mathcal T$.
In this case, the number of XBWT runs $r$ is equal to $12$, where every $c$-run break in the XBWT has been highlighted in red.
}
\label{fig: example XBWT}
\end{figure}

\subsection{The Aho-Corasick automaton}

In this subsection, we recall a result presented in the article of Kosolobov and Sivukhin~\cite{Kosolobov2019}, namely that their data structure for representing the XBWT of a trie, previously described in this section, can also be applied to represent the next transitions in the Aho-Corasick automaton~\cite{AC75}.
In fact, all the solutions presented in this section were originally proposed specifically to represent the trie $\mathcal T$ underlying the Aho-Corasick automaton, while in this article we presented them in a more general setting as representations of the XBWT.
As a direct consequence of the analysis of Theorem~\ref{theorem: trieHk id}, the space usage of a compressed representation of the Aho-Corasick automaton can be written in terms of our empirical entropy plus other terms (see Corollary~\ref{corollary: AC}).
We now start by briefly describing the Aho-Corasick automaton.
Given a text $T$ and a dictionary $S$ of $d$ strings, the Aho-Corasick automaton is designed to recognize which strings of $S$ appear in $T$, as well as their locations in the text.
Informally, the Aho-Corasick automaton can be described as the trie encoding the set of strings $S$, endowed with additional backward edges used to efficiently find all the occurrences of $S$ in $T$.
In this article, we only provide a high-level description of the algorithm execution, while a more detailed description can be found in the original article of Aho and Corasick~\cite{AC75}.
Consider the set $P$ formed by all the distinct prefixes in $S$, including the empty string $\epsilon$ and all the strings of $S$, and let $n$ be the cardinality of $P$.
Then the total number of states in the Aho-Corasick automaton is equal to $n$, and every state corresponds to a specific element of the set $P$.
In the following, as done by Belazzougui~\cite{BelazzouguiAhoCorasick}, we consider the variant of the Aho-Corasick automaton where every state $u$ has three types of transitions called \emph{next}, \emph{failure}, and \emph{report}.
Let $p \in P$ be the prefix associated to the state $u$, then these three types of transitions of $u$ are defined as follows:

\begin{itemize}
    \item For every character $c \in \Sigma$, if $pc \in P$, then the next transition \texttt{next}($u$,$c$) returns the state of the automaton associated to the prefix $pc$.
    \item Let $\bar p$ be the longest proper suffix of $p$ such that $\bar p \in P$ (possibly $p = \epsilon$), then the failure transition of $u$ \texttt{failure}($u$) returns the state associated to $\bar p$.
    \item Let $\bar p$ be the longest proper suffix of $p$ such that $\bar p \in S$.
    If $\bar p$ exists, then the state $u$ has a report transition \texttt{report}($u$) returning the accepting state associated to the prefix $\bar p$.
\end{itemize}

In addition to these transitions, the Aho-Corasick automaton requires additional data structures to recognise the accepting states, that is the states corresponding to the strings of $S$.
The initial state of the automaton is the one corresponding to $\epsilon$, and during the algorithm execution we scan the text from left to right.
Suppose that at some point of the algorithm execution we have reached the character $T[i]=c$ of the text and a state $u$ of the automaton corresponding to the prefix $p \in P$.
Then at the next iteration, the algorithm executes the following operations: (i)~If $pc$ is not a prefix of $P$, we follow the failure transition \texttt{failure}($u$).
(ii)~If $pc$ is a prefix of $P$, we follow the next transition \texttt{next}($u$, $c$), we check if we entered into an accepting state, and we advance one position in the text.
(iii)~If we have followed a next transition, we use the report transitions to find all the occurrences of $S$ in $T$ that are proper suffixes of the string $pc$ (if any).
It is known that the total number failure and next transitions we have to follow during the algorithm execution is bounded by $2 \lvert T \rvert$~\cite[Theorem 2]{AC75}.
Moreover, the number of times a report transition is followed is trivially bounded by the number of occurrences of the strings in $S$ within the text $T$, denoted hereafter by $occ$.
Thus, if we assume we can follow a transition of any type in $O(1)$ time, the total complexity of the algorithm becomes $O(\lvert T \rvert + occ)$.
It is easy to observe that these next transitions form a trie $\mathcal T$ of $n$ nodes rooted at the initial state of the automaton.
As a consequence, by Corollary~\ref{cor: prefix matching} it follows that the trie $\mathcal T$, and consequently the next transitions, can be encoded in at most $n\mathcal H_k (\mathcal  T) + o(n)$ bits, while supporting \texttt{next}($u$,$c$) in $O(1)$ time for every $c \in \Sigma$ and node $u$.

In order to represent the remaining parts of the Aho-Corasick automaton, that is, the failure transitions, the report transitions, and the accepting states, we consider the solutions proposed in the articles of Belazzougui~\cite{BelazzouguiAhoCorasick}, and Kosolobov and Sivukhin~\cite{Kosolobov2019}.
Specifically, to distinguish between non-accepting and accepting states, Belazzougui proposed to use the ID representation of Raman et al.~\cite{RRR}, which can be stored in at most $d(\log (n/d) + O(1))$ bits of space, where $d = \lvert S \rvert$~\cite[Section 3.1]{BelazzouguiAhoCorasick}.
The identifier of a string $s \in S$ and its accepting state, returned by this ID representation, corresponds to the co-lexicographic rank of $s$ among the prefixes in $P$.
Moreover, given the co-lexicographic rank of the strings $s$ in the set $P$, we can store the lengths of the strings in $S$ in an array of $d$ positions in $d(\log (m/d) + O(1) )$~\cite[Section 3.5]{BelazzouguiAhoCorasick}, where $m$ is the total length of the strings in $S$, formally $m = \sum_{s \in S} \lvert s \rvert$.
This latter result is achieved by using the solution of Grossi and Vitter~\cite{EFGrossiVitter}, which supports random accesses in $O(1)$ time.
For the report transitions, Belazzougui proposed a representation taking $d(\log(n/d) + O(1))$~\cite[Lemma 7]{BelazzouguiAhoCorasick} bits of space by using the compressed tree representation proposed by Jansson et al.~\cite{JanssonDegreeEntropy}, so that \texttt{report}($u$) can be executed in $O(1)$ time.
Finally, to encode the failure transitions we use the solution of Kosolobov and Sivukhin~\cite[Section 5]{Kosolobov2019} which allows us to save the failure transitions within $\beta n + o(n)$ bits, where $\beta$ is an arbitrary constant satisfying $0 < \beta < 2$.
Their solution works through a sampling scheme saving the failure transitions only for some nodes in the automaton.
This sampling technique introduces a slowdown in the Aho-Corasick algorithm, where the smaller is the constant $\beta$, the slower becomes the algorithm execution.
However, this slowdown is hidden by the big-O notation and consequently the time complexity remains $O(\lvert T \rvert + occ)$~\cite[Theorem 9]{Kosolobov2019}.
This machinery allows us to find all the occurrences of the strings in $S$ within the text $T$ through triples (\texttt{occ\_bgn}, \texttt{occ\_end}, \texttt{str\_id})~\cite[Section 3.5]{BelazzouguiAhoCorasick}, where \texttt{occ\_bgn} and \texttt{occ\_end} are the starting and ending positions of the occurrence in $T$, respectively, and \texttt{str\_id} is the unique identifier of the string $s \in S$ occurring in $T$, namely the co-lexicographic rank of $s$ in the set $P$.
By observing that $n \leq m+1$, since the number of distinct prefixes is upper-bounded by the total length of the strings in $S$, due to Corollary~\ref{cor: prefix matching}, we can deduce the following result.

\begin{corollary}\label{corollary: AC}
    Let $\mathcal T$ be the trie representing a set of strings $S$, and $\varepsilon$ be an arbitrary constant with $0 \leq \varepsilon <1$.
    If $\sigma \leq n^\varepsilon$, then we can represent the Aho-Corasick automaton corresponding to $S$ in at most $n\mathcal H_k(\mathcal T) + \beta n+  o(n) + d( 3\log (m/d) + O(1))$ bits of space, for every $k \leq \max\{0,\alpha\log_\sigma n -2 \}$ simultaneously, where $\alpha$ and $\beta$ are two constants satisfying $0 \leq \alpha < 1$ and $0 < \beta < 2$, respectively.
    This representation allows us to find all the occurrences of the strings in $S$ within a given text $T$ through triples (\texttt{occ\_bgn}, \texttt{occ\_end}, \texttt{str\_id}) in $O(\lvert T \rvert + occ)$ time.
\end{corollary}

\section{FM-index for tries}\label{sec: fmindex}

In this section, we aim to enhance the representation of the $\BWT(\mathcal T)$~\cite{Kosolobov2019} discussed in Section~\ref{section: XBWT} to efficiently support \emph{subpath queries}~\cite{xBWTjournal, rindexTrie}, also called count queries and denoted by $\countop(\mathcal T, p)$ in this paper.
This query consists in counting the number of nodes reached by path labeled by a string $p$, where $p \in \Sigma^m$ for some $m \geq 0$, that is $\countop(\mathcal{T},p)=\lvert\{u\in V \mid \lambda_m(u)=p \}\rvert$.
Note that the considered paths can start from any node of $\mathcal{T}$.
In this section we show that, the representation discussed in Theorem~\ref{theorem: trieHk id} supports $\countop(\mathcal{T},p)$ queries in $O(m(\log \sigma + \log \log n))$ time, where $m = \lvert p \rvert$.
Furthermore, we prove that if $\sigma = O(\log^c n)$ for some arbitrary constant $c > 0$, by applying the FID representation of Lemma~\ref{thm:RRR}, instead of the ID, we can still represent the $\BWT(\mathcal T)$ in at most $n \mathcal H_k(\mathcal T) + o(n)$ bits of space, for every $k$ sufficiently small simultaneously, and reduce the count query time to $O(m)$.
We recall that both solutions also support $child(i,c)$ in $O(1)$ time.
Moreover, we prove that if for every character $c \in \Sigma$ it holds that $n_c \leq n / 2$, both representations are \emph{succinct}, that is, the number of bits they use is upper-bounded by $\mathcal{H}^{wc}(\mathcal T) + o(\mathcal H^{wc}(\mathcal T ))$, where $\mathcal H^{wc}(\mathcal T)$ is the worst-case entropy for tries of Definition~\ref{def: worst-case}.
These aforementioned representations of the $\BWT(\mathcal T)$ represent natural extensions of the renowned FM-index~\cite{FMindex,FMI-Journal} from strings to tries, since they use an analogous strategy to simultaneously compress and index the input trie.

We now recall how the XBWT of Ferragina et al.~\cite{xBWTjournal} can be employed to perform count queries.
It is known that for every integer $m \geq 0$, and string $p \in \Sigma^m$, the nodes $u$ satisfying $\lambda_m(u) = p$ form an interval in the co-lexicographically sorted sequence $u_1, u_2, \ldots , u_n$ of the nodes~\cite{xBWTjournal}.
In other words, for every string $p$ there exists an interval $[i,j]$, with $1 \leq i, j \leq n$, such that $u_i, u_{i+1}, \ldots, u_j$ are those and only nodes satisfying $\lambda_m(u) = p$.
If no node $u$ of $\mathcal T$ satisfies $\lambda_m (u) = p$, then $i > j$ holds true.
The interval $[i,j]$ corresponding to $p$ can be computed recursively by means of {\em forward search} as follows.
Consider the string $p'$ and the symbol $c$ such that $p = p'c$ and let $[i',j']$ be the interval of the nodes reached by $p'$. It is known that $i = C[c] +  \rank(i' - 1, B_c) + 1$ and $j = C[c] + \rank(j',B_c)$~\cite{xBWTjournal}, where $B_c$ and $C$ are the bitvector and the array defined in Section~\ref{section: XBWT}.
The empty string $\epsilon$ represents the base case of this recursion, and its corresponding interval is $[1,n]$.
When the forward search has terminated, the integer $\countop(\mathcal T,p)$ can be trivially computed as $j - i + 1$ if $i \leq j$, and $0$ otherwise.
In the following theorem we show that the representation of $\BWT(\mathcal T)$ proposed by Kosolobov and Sivukhin~\cite{Kosolobov2019} and analysed in Theorem~\ref{theorem: trieHk id} supports $\countop(\mathcal T,p)$ queries in $O(m(\log \sigma+ \log \log n))$ time where $m=\lvert p\lvert$. Furthermore we show that if the alphabet is sufficiently small, this time complexity can be reduced to $O(m)$.
To prove this second result, we use the solution of Pătrașcu to perform rank and select queries on a given bitvector~\cite[Theorem 2]{succincter}.
Indeed, we show that by applying this solution to the $\sigma$ bitvectors $B_c$, without any explicit partitioning, we can automatically reach the desired upper-bound of $n\mathcal H_k(\mathcal T) + o(n)$ bits, which holds for every $k \geq 0$ sufficiently small simultaneously.

\begin{theorem}\label{theorem: trieHk fid}
    Let $\mathcal{T}$ be a trie and $\varepsilon,\alpha$ be arbitrary real constants with $0\leq \alpha,\varepsilon <1$.
    Then there exist encodings of $\BWT(\mathcal T)$ occupying at most $n\mathcal H_k(\mathcal T) + o(n)$ bits of space for every integer $k \leq \max\{0, \alpha \log _\sigma n - 2 \}$ simultaneously that given a pattern $p \in \Sigma^m$ support $\countop(\mathcal T, p)$ queries in 1)~$O(m(\log \sigma+ \log \log n))$ time if $\sigma \leq n^\varepsilon$ and 2)~$O(m)$ time when $\sigma = O(polylog(n))$.
\end{theorem}

\begin{proof}
1)~For general alphabets of size $\sigma\leq n^\varepsilon$, consider the representation of $\BWT(\mathcal T)$ proposed by Kosolobov and Sivukhin~\cite{Kosolobov2019} and analysed in Theorem~\ref{theorem: trieHk id}.
In Proposition~\ref{prop: full rank queries on id}, we proved that this data structure can support full rank queries in $O(\log \sigma + \log \log n)$ time, therefore it supports $\countop(\mathcal T, p)$ queries in $O(m (\log \sigma + \log \log n))$ time.
The claimed space bound directly follows from Theorem~\ref{theorem: trieHk id}.

2)~Otherwise, if $\sigma=O(polylog(n))$ consider the solution of Pătrașcu~\cite[Theorem 2]{succincter} described in~\cite[Subsection 4.1]{succincter}. Pătrașcu's representation partitions the input bitvector $B$ of size $n$ into blocks each of size $r=(\frac{\log n}{t})^{\Theta(t)}$ bits, and stores each block as a succinct augmented B-tree (aB-tree) within $\log \binom{r}{x^i}+2$ bits of space, where $x^i$ is the number of ones in the $i$-th block $B^i$, with $i\in [h]$ and $h=\lceil n/r\rceil$. 
Accounting for all the $h$ aB-trees and the auxiliary data structures needed to efficiently support the queries, the total space of this solution can be rewritten as $\sum_{i=1}^h \log \binom{r}{x^i}+n/(\frac{\log n}{t})^t + \tilde{O}(n^{3/4})$. 
We observe that $\log \binom{r}{x^i} \leq r\mathcal{H}_0(\tilde B^i)$~\cite[Equation 11.40]{elementsOfInformationTheory},  where $\tilde{B}$ is the bitvector $B$ right-padded by $0$s to make its length a multiple of $r$.
Therefore, the total space is upper-bounded by $\sum_{i=1}^h r\mathcal{H}_0(\tilde{B}^i)+n/(\frac{\log n}{t})^t + \tilde{O}(n^{3/4})$.
We now apply this FID representation to every right-padded bitvector $\tilde{B}_c$. 
The total space occupation is at most $\sum_{c\in \Sigma}\sum_{i=1}^h r\mathcal{H}_0(\tilde{B}^i_c)+\sigma n/(\frac{\log n}{t})^t + \tilde{O}(n^{3/4})$ bits, where the $\tilde{O}(n^{3/4})=o(n)$ term is the space for the shared global look-up tables of Pătrașcu's solution.
We note that since $\sigma = O(\log^c n)$ for some constant $c$, by choosing $t$ as a constant strictly larger than $c$, the term $\sigma n/(\frac{\log n}{t})^t$ becomes $O(n/\log ^{t-c}n)=o(n)$.
To prove the claimed space bound we upper-bound the above left summation in terms of $\mathcal{H}_k(\mathcal{T})$.
We preliminary observe that $\sum_{c\in \Sigma}\sum_{i=1}^h r\mathcal{H}_0(\tilde{B}^i_c) < \sum_{c\in \Sigma}\sum_{i=1}^h \lvert B_c^i\rvert \mathcal{H}_0(B^i_c) + \sigma r$, and that $\sigma r = o(n)$ since both $\sigma$ and $r$ are poly-logarithmic functions on $n$.
Next, we apply Lemma~\ref{lemma: partitionedH0} to the above FID partitioning and the optimal one, similarly as we did in the proof of Theorem~\ref{theorem: trieHk id}.
In particular, we obtain that $\sum_{c\in \Sigma} \sum_{i=1}^h \lvert B_{c}^i\rvert \mathcal{H}_0(B^i_c) \leq n\mathcal{H}_k(\mathcal{T}) + \sigma(\ell-1)r$, where $\ell \leq \sigma^k$ is the size of the context induced partitioning. 
If $k=0$, then $\ell=1$ and the last term $\sigma (\ell-1)r$ is $0$, otherwise it is at most $\sigma^{k+1} r \leq n^\alpha r = o(n)$ since $\alpha<1$ and $r$ is a poly-log.
This proves the upper-bound $n\mathcal{H}_k(\mathcal{T})+o(n)$ for every $k \leq \max\{0, \alpha \log _\sigma n - 1 \}$ simultaneously.
The above representation supports full rank queries on the $\sigma$ bitvectors in $O(1)$ time, because $t$ is a constant (see Lemma~\ref{thm:RRR}).
Consequently, this solution can support $\countop(\mathcal T, p)$ queries in $O(m)$ time for every string pattern $p$ of length $m$.
\end{proof}

We recall that both solutions support the operation $child(i, c)$ in $O(1)$ time, since partial rank queries are supported in constant time. 
We now show that if $n_c \leq n/2$ holds for every $c \in \Sigma$, then the data structures analysed in Theorems~\ref{theorem: trieHk id} and~\ref{theorem: trieHk fid} are \emph{succinct} with respect to the worst-case entropy for tries $\mathcal{H}^{wc}$ of Definition~\ref{def: worst-case}.

\begin{theorem}
    If $n_c \leq n/2$ holds for every $c \in \Sigma$, then the representations of $\BWT(\mathcal T)$ analysed in Theorems~\ref{theorem: trieHk id} and~\ref{theorem: trieHk fid} are succinct, i.e., they occupy a number of bits upper-bounded by $\mathcal{H}^{wc}(\mathcal T) + o(\mathcal H^{wc}(\mathcal T))$.
\end{theorem}

\begin{proof}
    To prove the claim, we need to show that $n\mathcal H_k(\mathcal T) + o(n)$ is upper-bounded by $\mathcal H^{wc}(\mathcal T) + o(\mathcal H^{wc}(\mathcal T))$.
    We preliminary show that $n\mathcal H_k(\mathcal T) \leq \mathcal H^{wc}(\mathcal T) + o(n)$ holds for every $k \geq 0$.
    By Corollary~\ref{corollary: relation worst-case empirical}, we know that $n\mathcal H_0(\mathcal T) = \mathcal H^{wc}(\mathcal T) + O(\sigma \log n)$, and consequently, due to the log-sum inequality~\cite[Eq. (2.99)]{elementsOfInformationTheory} we conclude that $n\mathcal H_k(\mathcal T) \leq \mathcal H^{wc}(\mathcal T) + O(\sigma \log n)$ for every $k \geq 0$.
    Moreover, in both theorems we assumed that $\sigma \leq n^{\varepsilon}$ for some constant $\varepsilon$, with $0 \leq \varepsilon < 1$.
    Consequently, it follows that $O(\sigma \log n) = o(n)$, which proves $n\mathcal H_k(\mathcal T) \leq \mathcal H^{wc}(\mathcal T) + o(n)$.
    Therefore, to complete the proof it remains to be shown that $o(n) = o(\mathcal H^{wc}(\mathcal T))$.
    To see this, we preliminary observe that the formula $\log \binom{m}{n}$ is lower-bounded by $n$ whenever $n \leq m/2$~\cite[Section 1.1.4]{RRR}.
    As a consequence of this, since $\mathcal{H}^{wc}(\mathcal T) =\sum_{c \in \Sigma} \log \binom{n}{n_c}  - \log n $ (see Definition~\ref{def: worst-case}), if $n_c \leq n/2$ holds for every $c \in \Sigma$, we can conclude that $\mathcal{H}^{\wc}(\mathcal{T})\geq n-1-\log n$.
    This implies that the $o(n)$ term is also $o(\mathcal{H}^{\wc}(\mathcal{T}))$, which completes the proof.
\end{proof}

On the other hand, if there exists a character $\hat c \in \Sigma$, such that $n_{\hat c} > n/2$, since $\sum_{c \in \Sigma} n_c = n-1$, then $\hat c$ is the only symbol of $\Sigma$ occurring more than $n/2$ times.
We also observe that $\mathcal{H}^{\wc}(\mathcal{T}) = \sum_{c \in \Sigma}\log \binom{n}{n_c} - \log n = \sum_{\substack{c \in \Sigma \setminus \{\hat c\}}} \log \binom{n}{n_c} + \log \binom{n}{n - n_{\hat c}} - \log n$ holds since $\binom{m}{n} = \binom{m}{m-n}$.
As $n-n_{\hat c}\leq n/2$, we observe that $\mathcal{H}^{\wc}(\mathcal{T})\geq n^*-\log n$ where $n^*=(n - n_{\hat c}) + \sum_{\substack{c \in \Sigma\setminus \{\hat c\}}} n_c$.
Now we replace the bitvector $B_{\hat c}$ with its complement $\bar{B}_{\hat c}$ obtained by replacing in $B_{\hat c}$ every $1$ with $0$ and vice versa.
Now we build the ID representation of Raman et al.~\cite[Theorem 4.6]{RRR} on every bitvector $B_c$ within $\sum_{\substack{c \in \Sigma \setminus \{\bar c\}}} \log \binom{n}{n_c} + \log \binom{n}{n - n_{\hat c}}+o(n^*)+O(\sigma\log\log n) = \mathcal H^{\wc}(\mathcal T) + \log n + o(n^*) + O(\sigma \log \log n)$ bits of space.
By adding the space usage for the array $C$ and the pointers to the $\sigma$ ID representations we obtain a total space usage of $\mathcal H^{\wc}(\mathcal T)  + o(n^*) + O(\sigma \log  n)$, that is $\mathcal H^{\wc}(\mathcal T)  + o(\mathcal H^{\wc}(\mathcal T)) + O(\sigma \log  n)$.
It follows that this representation is succinct up to an additive $O(\sigma \log n)$ number of bits.
Finally, assuming that alphabet is effective, i.e., that every $c\in \Sigma$ appears in $\mathcal{T}$, we observe that since $\log \binom{y}{x}\geq \log y$ if $x$ is not $0$ nor $y$, it holds that $\mathcal{H}^{\wc}(\mathcal{T}) =  \Omega( (\sigma-1) \log n)$, and therefore the previous space is also $O(\mathcal{H}^{\wc}(\mathcal{T}))$ assuming $\sigma > 1$.
By Remark~\ref{remark:IDrank}, we know that this representation can compute full rank queries on the $\sigma$ bitvectors in $O(\log n)$ time, and consequently $child(i,c)$ can be executed in $O(\log n)$ time. For the same reason $\countop(\mathcal T,p)$ can be computed in $O(m \log n)$ time, where $m = \lvert p \rvert$.

\subsection{Other operations}\label{subsec: other op}

We now describe the other operations supported by the data structures analysed in Theorems~\ref{theorem: trieHk id} and~\ref{theorem: trieHk fid}.
First of all, we consider the operation $\parentop(i)$.
Let $u_i$ be the $i$-th node in co-lexicographic order of $\mathcal T$.
The operation $\parentop (i)$ returns the co-lexicographic rank of the parent of $u$ in $\mathcal T$.
This operation is not defined if $u_i$ is the root of the trie, i.e., if it holds that $i = 1$.
For instance in Figure~\ref{fig: example XBWT}, we have that $\parentop(17) = 21$.
Due to the properties of the XBWT~\cite{xBWTjournal, xBWTconf}, we know that the integer $j$ satisfying $\parentop(i)$ can be computed as follows.
Consider the character $c = \lambda(u_i)$, then we have that $\parentop(i) = \select(i - C[c], B_c)$.
Consequently, the solutions discussed in Theorems~\ref{theorem: trieHk id} and~\ref{theorem: trieHk fid} can support this operation in constant time, since they both support select queries in constant time on the $\sigma $ bitvectors $B_c$.
Consider now another typical operation on ordered trees, which consists of finding the $j$-th child of an input node $u$ according to the total order of its children.
In the context of tries, this operation corresponds to finding the child $v$ of $u$ for which there exist exactly $j-1$ nodes $z$ such that (i)~$z$ is a child of $u$, and (ii)~$\lambda(z) < \lambda(v)$ holds.
Let $i$ be the co-lexicographic rank of $u$ and let $v$ be the $j$-th child of $u$.
If by hypothesis we know that $\bar c = \lambda(v)$, then the co-lexicographic rank of $v$ can be computed as $child(i,\bar c)$ in $O(1)$ time.
However, finding $\bar c$ appears to be the bottleneck of this operation.
We can retrieve this character by scanning the $\sigma$ bitvectors $B_c$ in order, returning the $j$-th character $c$ for which $B_{c}[i] = 1$ holds.
However, this solution implies a worst-case time complexity equal to $\Theta(\sigma)$.
We leave as an open problem how to improve the query time of this operation without exceeding the $n\mathcal H_k(\mathcal T) + o(n)$ spatial upper-bound.

\section{Comparison with the number r of XBWT runs}\label{sec: r index}

In this section we compare our empirical entropy measure $\mathcal H_k(\mathcal T)$ with the trie repetitiveness measure $r$ proposed by Prezza~\cite{rindexTrie} and defined as follows.
Consider the sequence $\BWT(\mathcal{T}) = out(u_1),\ldots,out(u_n)$ of Definition~\ref{definition: Burrows-Wheeler}.
An integer $i \in [n]$ is a $c$-run break if $c \in out(u_i)$ and either $i = n$ or $c\notin out(u_{i+1})$ holds. 
Then the number $r$ of \emph{XBWT runs}, is $r = \sum_{c \in \Sigma} r_c$, where $r_c$ is the total number of $c$-run breaks in $\BWT(\mathcal{T})$~\cite{rindexTrie} (see Figure~\ref{fig: example XBWT} for an example).
Prezza proved that every trie $\mathcal T$ can be represented in $O(r \log n)$ bits~\cite[Lemma 3.1]{rindexTrie}.
In the following, we observe an important relation between this number $r$ and $n\mathcal H_k(\mathcal T)$ of Definition~\ref{def: k-entropy}.
Note that by Corollary~\ref{corollary: relation worst-case empirical}, and due to the log-sum inequality, it follows that this result also implies a strong relation between the worst-case entropy $\mathcal H^{\wc}(\mathcal T)$ and $r$.
This result immediately follows from an upper-bound for $r$ proved by Prezza in his article~\cite[Theorem 3.1]{rindexTrie}.

\begin{corollary}\label{corollary:rHk}
    For every trie $\mathcal T$ and non-negative integer $k \geq 0$, the number $r$ of XBWT runs is upper-bounded by $n\mathcal H_k(\mathcal T) + \sigma^{k+1}$.
\end{corollary}

\begin{proof}
    In his article, Prezza proved that the number $r$ is upper-bounded by $\sum_{w \in \Sigma^{k}} \sum_{c \in \Sigma} \log \binom{n_{w}}{n_{w,c}} + \sigma^{k+1}$~\cite[Theorem 3.1]{rindexTrie}, where the integers $n_w$ and $n_{w,c}$ are those of Definition~\ref{def: integers nw and nwc}.
    However, by~\cite[Equation 11.40]{elementsOfInformationTheory} this in turn implies that $r \leq n\mathcal{H}_k(\mathcal{T}) +  \sigma^{k+1}$ for every $k \geq 0$.
\end{proof}

Note that a similar relation also holds for the case of strings.
Indeed, Mäkinen and Navarro proved that for every string $S$, it holds that its number of BWT runs $r$ is upper-bounded by $n\mathcal H_k(\mathcal S) + \sigma^{k}$, for every $k \geq 0$~\cite[Theorem 2]{makinenrupperboundJournal}, where $n\mathcal H_k(\mathcal S)$ is the $k$-th empirical order of the string $S$.
In the following proposition, we exhibit a family of trie for which $r = \Theta(n\mathcal{H}_0(\mathcal{T})) = \Theta(n)$ holds.

\begin{proposition}\label{prop: family of tries r versus entropy}
    There exists an infinite family of tries with $n$ nodes, for which $r=\Theta(n\mathcal{H}_0(\mathcal{T})) = \Theta(n)$.
\end{proposition}

\begin{proof}
Consider the complete and balanced binary trie $\mathcal{T}$ of $n > 1$ nodes over the alphabet $\Sigma = \{a, b\}$.
We observe that if $h$ is the height of $\mathcal{T}$, then the set of strings spelled in $\mathcal{T}$, considering also those reaching the internal nodes, is $\bigcup_{i=0}^h \Sigma^i$ and each of these strings reaches a single node.
In this case, we have $n_a = n_b = (n - 1)/2$ since half edges are labeled by the character $a$ and the remaining by the other character $b$.
Consequently, by Definition~\ref{def: 0-entropy}, we have that $n\mathcal{H}_0(\mathcal{T}) = (n-1)\log(\frac{2n}{n-1}) + (n+1)\log(\frac{2n}{n+1}) \approx 2n = \Theta(n)$.
Therefore, to prove the proposition, we show that $r = \Theta(n)$.
Consider the nodes $u_1,\ldots, u_n$ in co-lexicographic order, we have that $out(u_i)=\{a,b\}$ if $u_i$ is an internal node, and obviously $out(u_i)=\emptyset$ if $u_i$ is a leaf.
Since $n > 1$, node $u_1$ is the root and internal, while $u_n$ is necessarily the leaf reached by the string $b^{h}$.
Consequently, we have that $r$ is exactly twice the numbers of maximal leaf runs in $\BWT(\mathcal{T})$.
We now show that every maximal run of leaves contains exactly two nodes.
Consider the leaf $u_i$ reached by the string $a\alpha$ with $\alpha\in\Sigma^{h-1}$, clearly the next node $u_{i+1}$ is the leaf reached by the string $b\alpha$.
We can also observe that the node $u_{i-1}$ is necessarily the one reached by $\alpha$, which is internal since $\alpha \in \Sigma^{h-1}$. 
The string reaching node $u_{i+2}$ is $b\alpha[j+1..h-1]$ where $j$ is the position of the leftmost $a$ in $\alpha$.
Note that if $j$ does not exist then $\alpha=b^{h-1}$ and therefore $u_{i+1}$ is $u_n$, otherwise $b\alpha[j+1..h-1]$ has length at most $h-1$ and therefore $u_{i+2}$ is an internal node.
Since the number of leaves is $(n + 1)/2$, the number of leaf runs is $(n + 1)/4$ and therefore $r=\Theta(n)$.
\end{proof}

We observe that also this result is consistent with a known result in the realm of strings.
Indeed, it is known that the binary De Bruijn sequences $S$ have a number of BWT runs equal to $\Theta(n)$, where $n$ is length of the string~\cite{NavarroRipetitivenessMeasures, CompositeRepetitionAware}.
Note that since both $r \leq n\mathcal H_k(S) + \sigma^k$ and $n\mathcal{H}_0(S) \leq n$ hold, this in turn implies that $r = \Theta(n\mathcal H_0(\mathcal T)) = \Theta(n)$ holds for this family of strings.

\subsection{Comparison with the trie r-index}

We now compare the $r$-index for tries proposed by Prezza~\cite{rindexTrie} with those discussed in Theorems~\ref{theorem: trieHk id} and~\ref{theorem: trieHk fid}.
In this subsection, we assume that the trie alphabet is effective, namely that every character of the alphabet labels at least an edge of the trie. 
This implies that both $\sigma-1 \leq r$ and $\sigma \leq n$ hold. Moreover, given a string $S\in \Sigma^*$ and a symbol $c\in \Sigma$ we consider the queries $rank_c(i,S)$ returning the number of occurrences of $c$ in $S$ up to position $i$ included, and $select_c(i,S)$ returning the position of the $i$-th left-to-right occurrence of $c$ in $S$.
The $r$-index for tries, is a generalization of the famous $r$-index for strings proposed by Gagie et al.~\cite{rindexString}. 
The trie $r$-index occupies $O(r \log n) + o(n)$\footnote{This is the space usage to represent the trie and support the operations considered in this section} bits, where $r$ is the number of \emph{XBWT runs} previously defined in this section.
Concerning the space usage, despite the upper-bound of Corollary~\ref{corollary:rHk}, we observe that due to Proposition~\ref{prop: family of tries r versus entropy}, the $r$-index of Prezza may use \emph{asymptotically} more space than the data structures of Theorems~\ref{theorem: trieHk id} and~\ref{theorem: trieHk fid} due to the $\log n$ multiplicative factor in the space occupation of the $r$-index.
We now recall some operations that are supported by the trie $r$-index in $O(r \log n ) + o(n)$ bits of space~\cite{rindexTrie}.
The trie $r$-index supports $\countop(\mathcal T, p)$ queries, by means of computing the co-lexicographic range $[i,j]$ of nodes reached by $p\in \Sigma^m$, in $O(m \log \log \sigma)$ time, and that this time can be reduced to $O(m)$ in the case of poly-logarithmic alphabets. 
Indeed, the operation described in~\cite[Lemma 4.5]{rindexTrie} also includes reporting the pre-order index of first node in the interval $[i,j]$.
However, if we are only interested in computing $i$ and $j$, then the additional succinct trie topology is not needed, and furthermore the $O(m\log \sigma)$ time complexity, reported in the article of Prezza~\cite{rindexTrie}, can be easily improved as follows.
Executing a forward step essentially boils down to count how many times a symbol $c\in \Sigma$ occurs in the sequence of sets $out(1),\ldots,out(i)$\footnote{The possibility of improving the running time of this operation was already mentioned by Prezza in the original article~\cite{rindexTrie}} for any $1\leq i \leq n$ (see Section~\ref{sec: fmindex}).
In~\cite[Lemma 4.5]{rindexTrie} Prezza showed how this operation can be reduced to answering a constant number of rank and select queries within $O(r\log n)+o(n)$ bits of space over 1)~a compressed bitvector $B$ marking the first node in every XBWT block, and 2)~a sequence $S'$ of length $n'=O(r)$, over a new alphabet $\Sigma'$ of size $\sigma'=O(\sigma)$. 
In his solution, $B$ is stored with a FID representation and thus supports both rank and select queries in $O(1)$ time.
On the other hand, the string $S'$ is represented as a wavelet tree~\cite{wavelettree,NavaWT}, which can answer these queries in $O(\log \sigma)$ time and represents the bottleneck of the $\countop(\mathcal T, p)$ operation. 
The time to execute $\countop(\mathcal T, p)$ can be easily improved to $O(m\log\log \sigma)$ time by representing $S'$ using the solution of~\cite[Theorem 5.2]{optSequenceRankSelect}. 
Indeed, in this way we can answer rank and select queries in $O(\log\log\sigma)$ and $O(1)$ time respectively, while using $n'\mathcal{H}_0(S')+o(n'\mathcal{H}_0(S'))+o(n')\subseteq O(r\log n)$ bits of space, since $n'=O(r)$ and $\mathcal{H}_0(S')\leq \log \sigma'=O(\log n)$.
We observe that this solution requires that $\sigma'\leq n'$.
To ensure this, we can easily make the considered alphabet effective by storing a bitvector $I$ of $\sigma'=O(r)$ positions such that $I[i] = 1$, if and only if the $i$-th character of $\Sigma'$ appears in $S'$. Using $I$ and the previous data structure we can trivially execute rank and select for all the symbols of $\Sigma'$ in the claimed time and space.
Instead, when $\sigma = O(polylog(n))$, using the solutions described in~\cite[Theorem 4.1]{optSequenceRankSelect} or the one in~\cite[Theorem 3.2]{FullTextIndexes}, we can represent $S'$ within $O(r\log n)$ bits of space and answer both rank and select queries in $O(1)$ time and therefore a $\countop(\mathcal T, p)$ query in $O(m)$ total time.
Using these representations, the $r$-index can also execute $child(i,c)$ operations in $O(\log\log \sigma)$ and $O(1)$ time, respectively, by executing a single forward step from position $i$ with symbol $c$ and then checking whether $j<i$ holds or not.
Moreover, the $r$-index can retrieve the character labeling the $k$-th edge outgoing from a node $u$, given its co-lexicographic rank $i$, in $O(\log^2 \sigma)$ time~\cite[Section 4]{rindexTrie}.
This, combined with the fact it can execute child queries in $O(\log \log \sigma)$, implies that it can also access the $k$-th child of any node in $O(\log^2 \sigma)$ time.
On the other hand, we recall that the solutions considered in Theorems~\ref{theorem: trieHk id} and~\ref{theorem: trieHk fid} support $child(i,c)$ operations in $O(1)$ time, $\countop(\mathcal T, p)$ operations in $O(m(\log \sigma+ \log \log n))$ time, for alphabets of size $\sigma \leq n^\varepsilon$ with $0\leq \epsilon <1$, and $O(m)$ time when $\sigma = O(polylog(n))$.
Furthermore, they can access the $k$-th child of a node in $O(\sigma)$ as discussed in Subsection~\ref{subsec: other op}.

\section{Conclusions}\label{sec: conclusions}

In this paper, we considered the problem of measuring the amount of information within tries.
We started by giving an alternative proof for the closed formula counting the number of tries with a given symbol distribution.
We used this formula to propose a new worst-case entropy $\mathcal H^{\wc}$ for tries, which, by exploiting the information concerning this label distribution, refines the standard worst-case entropy $\mathcal{C}(n, \sigma)$ which considers the tries with $n$ nodes over an alphabet of size $\sigma$.
After that, we introduced an empirical entropy $\mathcal H_k$, specifically designed for tries, which considers the symbols statistics, refined according to the context of the nodes.
Throughout the article we showed that many known results for strings, naturally extends to our trie entropies.
Indeed, we showed that the relation between $n\mathcal H_0$ and $\mathcal H^{\wc}$ is analogous to the one between the $0$-th order empirical entropy for strings and the worst-case string entropy.
Then, we showed that the renowned compression algorithm arithmetic coding can be extended to tries for the same purpose, that is to store the input trie within $n\mathcal H_k(\mathcal T) + 2$ bits of space, plus the space required to save the (conditional) probabilities.
We also compared our empirical trie entropy with the tree label entropy $\mathcal{H}^{label}_k$ of Ferragina et al.~\cite{xBWTconf} which measures the compressibility of the labels of an ordered tree.
We showed that, for the specific case of tries, the inequality $n\mathcal{H}_k(\mathcal{T})\leq (n-1)\mathcal{H}^{label}_k(\mathcal{T})+1.443n$ holds for every trie $\mathcal T$.
Interestingly, the $1.443n$ term is smaller than the $2n - \Theta(\log n)$ additional bits required in the worst case to store the trie topology, as an unlabeled ordered tree, separately from its labels. 
Furthermore, we showed that in some cases $n\mathcal{H}_k(\mathcal{T})$ can be asymptotically smaller than the label contribution $(n-1)\mathcal{H}^{label}_k(\mathcal{T})$ alone.
In the remaining part of the article we proved the existence of representations for the XBWT~\cite{xBWTjournal} of a trie, that can be stored within $n\mathcal H_k(\mathcal T) + o(n)$ bits of space, for every $k$ sufficiently small simultaneously.
Note that also in this case there exists a similar result for the case of strings, since the Burrows-Wheeler transform of a string $S$ can be compressed to $\lvert S \rvert \mathcal H_k(S)$ bits, plus other terms, of space~\cite{FullTextIndexes}, with $\mathcal H_k(S)$ being the empirical $k$-th order (string) entropy of $S$.
We also showed that these representations can be regarded as a natural extension of the FM-index~\cite{FMI-Journal} from strings to tries, as they support count queries and other operations directly on the entropy-compressed XBWT.
Finally, we compared the number $r$ of XBWT runs~\cite{rindexTrie} and our empirical entropy and we showed that their relation is similar to the one known for strings.
Indeed, $r$ is always upper-bounded by $n\mathcal H_k(\mathcal T) + \sigma^{k+1}$ and for some family of tries the relation $r = \Theta(\mathcal H_0(\mathcal T)) = \Theta(n)$ holds.

\bibliographystyle{plainurl}

\begin{thebibliography}{10}

\bibitem{AC75}
Alfred~V. Aho and Margaret~J. Corasick.
\newblock Efficient string matching: an aid to bibliographic search.
\newblock {\em Commun. ACM}, 18(6):333–340, June 1975.
\newblock \href {https://doi.org/10.1145/360825.360855}
  {\path{doi:10.1145/360825.360855}}.

\bibitem{ArroyueloDynamicTries}
Diego Arroyuelo, Pooya Davoodi, and Srinivasa~Rao Satti.
\newblock Succinct dynamic cardinal trees.
\newblock {\em Algorithmica}, 74(2):742–777, 2016.
\newblock \href {https://doi.org/10.1007/s00453-015-9969-x}
  {\path{doi:10.1007/s00453-015-9969-x}}.

\bibitem{BelazzouguiAhoCorasick}
Djamal Belazzougui.
\newblock Succinct dictionary matching with no slowdown.
\newblock In Amihood Amir and Laxmi Parida, editors, {\em Combinatorial Pattern
  Matching, 21st Annual Symposium, {CPM} 2010, New York, NY, USA, June 21-23,
  2010. Proceedings}, volume 6129 of {\em Lecture Notes in Computer Science},
  pages 88--100. Springer, 2010.
\newblock \href {https://doi.org/10.1007/978-3-642-13509-5\_9}
  {\path{doi:10.1007/978-3-642-13509-5\_9}}.

\bibitem{CompositeRepetitionAware}
Djamal Belazzougui, Fabio Cunial, Travis Gagie, Nicola Prezza, and Mathieu
  Raffinot.
\newblock Composite repetition-aware data structures.
\newblock In Ferdinando Cicalese, Ely Porat, and Ugo Vaccaro, editors, {\em
  Combinatorial Pattern Matching - 26th Annual Symposium, {CPM} 2015, Ischia
  Island, Italy, June 29 - July 1, 2015, Proceedings}, volume 9133 of {\em
  Lecture Notes in Computer Science}, pages 26--39. Springer, 2015.
\newblock \href {https://doi.org/10.1007/978-3-319-19929-0\_3}
  {\path{doi:10.1007/978-3-319-19929-0\_3}}.

\bibitem{optSequenceRankSelect}
Djamal Belazzougui and Gonzalo Navarro.
\newblock Optimal lower and upper bounds for representing sequences.
\newblock {\em {ACM} Trans. Algorithms}, 11(4):31:1--31:21, 2015.
\newblock \href {https://doi.org/10.1145/2629339} {\path{doi:10.1145/2629339}}.

\bibitem{RepresentingTrees}
David Benoit, Erik~D. Demaine, J.~Ian Munro, Rajeev Raman, Venkatesh Raman, and
  S.~Srinivasa Rao.
\newblock Representing trees of higher degree.
\newblock {\em Algorithmica}, 43(4):275--292, 2005.
\newblock \href {https://doi.org/10.1007/S00453-004-1146-6}
  {\path{doi:10.1007/S00453-004-1146-6}}.

\bibitem{originalArticleTries}
Rene De~La Briandais.
\newblock File searching using variable length keys.
\newblock In R.~R. Johnson, editor, {\em Papers presented at the the 1959
  western joint computer conference, {IRE-AIEE-ACM} 1959 (Western), San
  Francisco, California, USA, March 3-5, 1959}, pages 295--298. {ACM}, 1959.
\newblock \href {https://doi.org/10.1145/1457838.1457895}
  {\path{doi:10.1145/1457838.1457895}}.

\bibitem{BWT}
Michael Burrows and David Wheeler.
\newblock A block-sorting lossless data compression algorithm.
\newblock {\em SRS Research Report}, 124, 1994.

\bibitem{elementsOfInformationTheory}
Thomas~M. Cover and Joy~A. Thomas.
\newblock {\em Elements of Information Theory}.
\newblock Wiley-Interscience, USA, 2006.

\bibitem{Bonsai}
John~J. Darragh, John~G. Cleary, and Ian~H. Witten.
\newblock Bonsai: a compact representation of trees.
\newblock {\em Softw. Pract. Exper.}, 23(3):277–291, 1993.
\newblock \href {https://doi.org/10.1002/spe.4380230305}
  {\path{doi:10.1002/spe.4380230305}}.

\bibitem{DavoodiDynamicTries}
Pooya Davoodi and Satti~Srinivasa Rao.
\newblock Succinct dynamic cardinal trees with constant time operations for
  small alphabet.
\newblock In Mitsunori Ogihara and Jun Tarui, editors, {\em Theory and
  Applications of Models of Computation}, pages 195--205, Berlin, Heidelberg,
  2011. Springer Berlin Heidelberg.
\newblock \href {https://doi.org/10.1007/978-3-642-20877-5_21}
  {\path{doi:10.1007/978-3-642-20877-5_21}}.

\bibitem{cycleLemma}
Nachum Dershowitz and Shmuel Zaks.
\newblock The cycle lemma and some applications.
\newblock {\em Eur. J. Comb.}, 11(1):35--40, 1990.
\newblock \href {https://doi.org/10.1016/S0195-6698(13)80053-4}
  {\path{doi:10.1016/S0195-6698(13)80053-4}}.

\bibitem{aProblemOfArrangements}
A.~Dvoretzky and Th. Motzkin.
\newblock {A problem of arrangements}.
\newblock {\em Duke Mathematical Journal}, 14(2):305 -- 313, 1947.
\newblock \href {https://doi.org/10.1215/S0012-7094-47-01423-3}
  {\path{doi:10.1215/S0012-7094-47-01423-3}}.

\bibitem{UniversalSuccinctTree?}
Arash Farzan, Rajeev Raman, and S.~Srinivasa Rao.
\newblock Universal succinct representations of trees?
\newblock In Susanne Albers, Alberto Marchetti{-}Spaccamela, Yossi Matias,
  Sotiris~E. Nikoletseas, and Wolfgang Thomas, editors, {\em Automata,
  Languages and Programming, 36th International Colloquium, {ICALP} 2009,
  Rhodes, Greece, July 5-12, 2009, Proceedings, Part {I}}, volume 5555 of {\em
  Lecture Notes in Computer Science}, pages 451--462. Springer, 2009.
\newblock \href {https://doi.org/10.1007/978-3-642-02927-1\_38}
  {\path{doi:10.1007/978-3-642-02927-1\_38}}.

\bibitem{magicAlgorithms}
Paolo Ferragina.
\newblock {\em Pearls of Algorithm Engineering}.
\newblock Cambridge University Press, 2023.

\bibitem{xBWTconf}
Paolo Ferragina, Fabrizio Luccio, Giovanni Manzini, and S.~Muthukrishnan.
\newblock Structuring labeled trees for optimal succinctness, and beyond.
\newblock In {\em 46th Annual {IEEE} Symposium on Foundations of Computer
  Science {(FOCS} 2005), 23-25 October 2005, Pittsburgh, PA, USA, Proceedings},
  pages 184--196. {IEEE} Computer Society, 2005.
\newblock \href {https://doi.org/10.1109/SFCS.2005.69}
  {\path{doi:10.1109/SFCS.2005.69}}.

\bibitem{xBWTjournal}
Paolo Ferragina, Fabrizio Luccio, Giovanni Manzini, and S.~Muthukrishnan.
\newblock Compressing and indexing labeled trees, with applications.
\newblock {\em J. {ACM}}, 57(1):4:1--4:33, 2009.
\newblock \href {https://doi.org/10.1145/1613676.1613680}
  {\path{doi:10.1145/1613676.1613680}}.

\bibitem{FMindex}
Paolo Ferragina and Giovanni Manzini.
\newblock Opportunistic data structures with applications.
\newblock In {\em 41st Annual Symposium on Foundations of Computer Science,
  {FOCS} 2000, 12-14 November 2000, Redondo Beach, California, {USA}}, pages
  390--398. {IEEE} Computer Society, 2000.
\newblock \href {https://doi.org/10.1109/SFCS.2000.892127}
  {\path{doi:10.1109/SFCS.2000.892127}}.

\bibitem{FMI-Journal}
Paolo Ferragina and Giovanni Manzini.
\newblock Indexing compressed text.
\newblock {\em J. ACM}, 52(4):552–581, 2005.
\newblock \href {https://doi.org/10.1145/1082036.1082039}
  {\path{doi:10.1145/1082036.1082039}}.

\bibitem{FullTextIndexes}
Paolo Ferragina, Giovanni Manzini, Veli M\"{a}kinen, and Gonzalo Navarro.
\newblock Compressed representations of sequences and full-text indexes.
\newblock {\em ACM Trans. Algorithms}, 3(2):20–es, 2007.
\newblock \href {https://doi.org/10.1145/1240233.1240243}
  {\path{doi:10.1145/1240233.1240243}}.

\bibitem{analyticCombinatorics}
Philippe Flajolet and Robert Sedgewick.
\newblock {\em Analytic Combinatorics}.
\newblock Cambridge University Press, 2009.

\bibitem{rindexString}
Travis Gagie, Gonzalo Navarro, and Nicola Prezza.
\newblock Fully functional suffix trees and optimal text searching in bwt-runs
  bounded space.
\newblock {\em J. {ACM}}, 67(1):2:1--2:54, 2020.
\newblock \href {https://doi.org/10.1145/3375890} {\path{doi:10.1145/3375890}}.

\bibitem{FixBlockCompressionBoostJournal}
Simon Gog, Juha K\"{a}rkk\"{a}inen, Dominik Kempa, Matthias Petri, and Simon~J.
  Puglisi.
\newblock Fixed block compression boosting in {FM}-indexes: Theory and
  practice.
\newblock {\em Algorithmica}, 81(4):1370–1391, April 2019.
\newblock \href {https://doi.org/10.1007/s00453-018-0475-9}
  {\path{doi:10.1007/s00453-018-0475-9}}.

\bibitem{concreteMathematics}
Ronald~L. Graham, Donald~E. Knuth, and Oren Patashnik.
\newblock {\em Concrete Mathematics: {A} Foundation for Computer Science, 2nd
  Ed}.
\newblock Addison-Wesley, 1994.
\newblock URL: \url{https://www-cs-faculty.stanford.edu/\%7Eknuth/gkp.html}.

\bibitem{wavelettree}
Roberto Grossi, Ankur Gupta, and Jeffrey~Scott Vitter.
\newblock High-order entropy-compressed text indexes.
\newblock In {\em Proceedings of the Fourteenth Annual {ACM-SIAM} Symposium on
  Discrete Algorithms, January 12-14, 2003, Baltimore, Maryland, {USA}}, pages
  841--850. {ACM/SIAM}, 2003.
\newblock URL: \url{http://dl.acm.org/citation.cfm?id=644108.644250}.

\bibitem{EFGrossiVitter}
Roberto Grossi and Jeffrey~Scott Vitter.
\newblock Compressed suffix arrays and suffix trees with applications to text
  indexing and string matching (extended abstract).
\newblock In F.~Frances Yao and Eugene~M. Luks, editors, {\em Proceedings of
  the Thirty-Second Annual {ACM} Symposium on Theory of Computing, May 21-23,
  2000, Portland, OR, {USA}}, pages 397--406. {ACM}, 2000.
\newblock \href {https://doi.org/10.1145/335305.335351}
  {\path{doi:10.1145/335305.335351}}.

\bibitem{HonTrieLabelEntropy}
Wing-Kai Hon, Tsung-Han Ku, Rahul Shah, Sharma~V. Thankachan, and Jeffrey~Scott
  Vitter.
\newblock Faster compressed dictionary matching.
\newblock {\em Theoretical Computer Science}, 475:113--119, 2013.
\newblock \href {https://doi.org/10.1016/j.tcs.2012.10.050}
  {\path{doi:10.1016/j.tcs.2012.10.050}}.

\bibitem{HuffmanCoding}
David~A. Huffman.
\newblock A method for the construction of minimum-redundancy codes.
\newblock {\em Proceedings of the IRE}, 40(9):1098--1101, 1952.
\newblock \href {https://doi.org/10.1109/JRPROC.1952.273898}
  {\path{doi:10.1109/JRPROC.1952.273898}}.

\bibitem{JanssonDegreeEntropy}
Jesper Jansson, Kunihiko Sadakane, and Wing{-}Kin Sung.
\newblock Ultra-succinct representation of ordered trees with applications.
\newblock {\em J. Comput. Syst. Sci.}, 78(2):619--631, 2012.
\newblock \href {https://doi.org/10.1016/J.JCSS.2011.09.002}
  {\path{doi:10.1016/J.JCSS.2011.09.002}}.

\bibitem{FixBlockCompressionBoostConf}
Juha K\"{a}rkk\"{a}inen and Simon~J. Puglisi.
\newblock Fixed block compression boosting in fm-indexes.
\newblock In {\em Proceedings of the 18th International Conference on String
  Processing and Information Retrieval}, SPIRE'11, page 174–184, Berlin,
  Heidelberg, 2011. Springer-Verlag.

\bibitem{Kosolobov2019}
Dmitry Kosolobov and Nikita Sivukhin.
\newblock Compressed multiple pattern matching.
\newblock In Nadia Pisanti and Solon~P. Pissis, editors, {\em 30th Annual
  Symposium on Combinatorial Pattern Matching, {CPM} 2019, June 18-20, 2019,
  Pisa, Italy}, volume 128 of {\em LIPIcs}, pages 13:1--13:14. Schloss Dagstuhl
  - Leibniz-Zentrum f{\"{u}}r Informatik, 2019.
\newblock \href {https://doi.org/10.4230/LIPICS.CPM.2019.13}
  {\path{doi:10.4230/LIPICS.CPM.2019.13}}.

\bibitem{makinenrupperboundJournal}
Veli M{\"{a}}kinen and Gonzalo Navarro.
\newblock Succinct suffix arrays based on run-length encoding.
\newblock {\em Nord. J. Comput.}, 12(1):40--66, 2005.

\bibitem{NavaWT}
Gonzalo Navarro.
\newblock Wavelet trees for all.
\newblock {\em Journal of Discrete Algorithms}, 25:2--20, 2014.
\newblock 23rd Annual Symposium on Combinatorial Pattern Matching.
\newblock \href {https://doi.org/10.1016/j.jda.2013.07.004}
  {\path{doi:10.1016/j.jda.2013.07.004}}.

\bibitem{compactDataStructures}
Gonzalo Navarro.
\newblock {\em Compact Data Structures - {A} Practical Approach}.
\newblock Cambridge University Press, 2016.

\bibitem{NavarroRipetitivenessMeasures}
Gonzalo Navarro.
\newblock Indexing highly repetitive string collections, part {I:}
  repetitiveness measures.
\newblock {\em {ACM} Comput. Surv.}, 54(2):29:1--29:31, 2022.
\newblock \href {https://doi.org/10.1145/3434399} {\path{doi:10.1145/3434399}}.

\bibitem{FullTextIndexesSurvey}
Gonzalo Navarro and Veli M\"{a}kinen.
\newblock Compressed full-text indexes.
\newblock {\em ACM Comput. Surv.}, 39(1):2–es, 2007.
\newblock \href {https://doi.org/10.1145/1216370.1216372}
  {\path{doi:10.1145/1216370.1216372}}.

\bibitem{succincter}
Mihai P{\u{a}}tra{\c{s}}cu.
\newblock Succincter.
\newblock In {\em 49th Annual {IEEE} Symposium on Foundations of Computer
  Science, {FOCS} 2008, Philadelphia, PA, USA, October 25-28, 2008}, pages
  305--313. {IEEE} Computer Society, 2008.
\newblock \href {https://doi.org/10.1109/FOCS.2008.83}
  {\path{doi:10.1109/FOCS.2008.83}}.

\bibitem{rindexTrie}
Nicola Prezza.
\newblock On locating paths in compressed tries.
\newblock In D{\'{a}}niel Marx, editor, {\em Proceedings of the 2021 {ACM-SIAM}
  Symposium on Discrete Algorithms, {SODA} 2021, Virtual Conference, January 10
  - 13, 2021}, pages 744--760. {SIAM}, 2021.
\newblock \href {https://doi.org/10.1137/1.9781611976465.47}
  {\path{doi:10.1137/1.9781611976465.47}}.

\bibitem{TRProdinger}
Helmut Prodinger.
\newblock Counting edges according to edge-type in $t$-ary trees, 2022.
\newblock \href {https://arxiv.org/abs/2205.13374} {\path{arXiv:2205.13374}},
  \href {https://doi.org/10.48550/arXiv.2205.13374}
  {\path{doi:10.48550/arXiv.2205.13374}}.

\bibitem{RRR}
Rajeev Raman, Venkatesh Raman, and Srinivasa~Rao Satti.
\newblock Succinct indexable dictionaries with applications to encoding k-ary
  trees, prefix sums and multisets.
\newblock {\em ACM Trans. Algorithms}, 3(4):43–es, 2007.
\newblock \href {https://doi.org/10.1145/1290672.1290680}
  {\path{doi:10.1145/1290672.1290680}}.

\bibitem{roteCountingDegrees}
G{\"u}nter Rote.
\newblock Binary trees having a given number of nodes with 0, 1, and 2
  children.
\newblock {\em S{\'e}minaire Lotharingien de Combinatoire, B38b}, 1997.
\newblock URL: \url{https://zbmath.org/0886.05081}.

\bibitem{SuffixTrees}
Peter Weiner.
\newblock Linear pattern matching algorithms.
\newblock In {\em 14th Annual Symposium on Switching and Automata Theory, Iowa
  City, Iowa, USA, October 15-17, 1973}, pages 1--11. {IEEE} Computer Society,
  1973.
\newblock \href {https://doi.org/10.1109/SWAT.1973.13}
  {\path{doi:10.1109/SWAT.1973.13}}.

\end{thebibliography}

\end{document}